\pdfoutput=1
\documentclass[acmsmall]{acmart}

\AtBeginDocument{%
  \providecommand\BibTeX{{%
    \normalfont B\kern-0.5em{\scshape i\kern-0.25em b}\kern-0.8em\TeX}}}

\usepackage{geometry}                		
\usepackage{graphicx}				
\usepackage{amsmath}
\usepackage[ruled,vlined,algosection]{algorithm2e}
\usepackage{xcolor}
\usepackage{hyperref}
\usepackage{makecell}
\usepackage{caption}
\usepackage{subcaption}
\setcopyright{none}

\usepackage{amsthm}
\newtheorem{claim}{Claim}
\newtheorem{condition}{Condition}
\newtheorem{remark}{Remark}[section]

\renewcommand{\P}{\mathbf P}
\newcommand{\E}{\mathbf{E}}

\begin{document}

\title{Source Detection via Contact Tracing in the Presence of Asymptomatic Patients}

\author{Gergely \'Odor}
\email{gergely.odor@epfl.ch}
\orcid{0000-0001-9139-249X}
\affiliation{%
  \institution{EPFL}
  \city{Lausanne}
  \country{Switzerland}
}

\author{Jana Vuckovic}
\email{jana.vuckovic@epfl.ch}
\affiliation{%
  \institution{EPFL}
  \city{Lausanne}
  \country{Switzerland}
}

\author{Miguel-Angel Sanchez Ndoye}
\email{miguel-angel.sanchezndoye@epfl.ch}
\affiliation{%
  \institution{EPFL}
  \city{Lausanne}
  \country{Switzerland}
}

\author{Patrick Thiran}
\email{patrick.thiran@epfl.ch}
\affiliation{%
  \institution{EPFL}
  \city{Lausanne}
  \country{Switzerland}
}

\begin{abstract}
Inferring the source of a diffusion in a large network of agents is a difficult but feasible task, if a few agents act as sensors revealing the time at which they got hit by the diffusion. One of the main limitations of current source detection algorithms is that they assume full knowledge of the contact network, which is rarely the case, especially for epidemics, where the source is called patient zero. Inspired by recent implementations of contact tracing algorithms, we propose a new framework, which we call Source Detection via Contact Tracing Framework (SDCTF). In the SDCTF, the source detection task starts at the time of the first hospitalization, and initially we have no knowledge about the contact network other than the identity of the first hospitalized agent. We may then explore the network by contact queries, and obtain symptom onset times by test queries in an adaptive way, i.e., both contact and test queries can depend on the outcome of previous queries. We also assume that some of the agents may be asymptomatic, and therefore cannot reveal their symptom onset time. Our goal is to find patient zero with as few contact and test queries as possible. We propose two local search algorithms for the SDCTF: the LS algorithm is more data-efficient, but can fail to find the true source if many asymptomatic agents are present, whereas the LS+ algorithm is more robust to asymptomatic agents. By simulations we show that both LS and LS+ outperform state of the art adaptive and non-adaptive source detection algorithms adapted to the SDCTF, even though these baseline algorithms have full access to the contact network. Extending the theory of random exponential trees, we analytically approximate the probability of success of the LS/ LS+ algorithms, and we show that our analytic results match the simulations. Finally, we benchmark our algorithms on the Data-driven COVID-19 Simulator (DCS) developed by Lorch et al., which is the first time source detection algorithms are tested on such a complex dataset.
    \end{abstract}
    

\maketitle

\section{Introduction}

During the COVID-19 pandemic, we have seen a revolution of the contact tracing technology, which helped track and contain the epidemic \cite{braithwaite2020automated,kretzschmar2020impact}. Some contact tracing programs were conducted by governmental/health agencies \cite{park2020contact}, while others relied on decentralized approaches \cite{troncoso2020decentralized}. Most contact tracing approaches work by notifying people who could have received the infection from known infectious patients, i.e., they trace ``forward'' in time. However, some advocate that a ``bidirectional'' tracing, where the past history of the infection is also tracked, can be more effective \cite{bradshaw2021bidirectional,endo2020implication,kojaku2021effectiveness}. In this paper we focus on the ``backward'' direction of the problem; the task of identifying the first patient who carried the disease, also called patient zero, or the source of the epidemic. The identification of patient zero can either be limited to a smaller population cluster, in which case it can be a first step towards ``bidirectional'' tracing, or it can be more ambitious; finding the first patient who developed the mutation of a certain disease can help understanding how the mutation occurred, which can help us prevent, or better prepare for future epidemics.

Surprisingly, given the importance of the problem and the relatively large literature on the topic, we are not aware of any instance where source detection algorithms have been applied in real situations, including during the COVID-19 pandemic. Our goal in this paper is to examine the applicability of the source detection models in the literature (which we call frameworks from now on), and then propose a new framework, which improves them in several aspects. Originally, source detection was introduced in the context of rumor spreading instead of epidemics by Zaman and Shah in their pioneering Sigmetrics paper~\cite{shah2010rumors,shah2011rumors}. Translating to the language of epidemics for clarity, in the framework of \cite{shah2011rumors}, an epidemic spreads over a network of agents that is completely known to us, and we observe a \textit{snapshot} of the network, which means that every agent reveals if they are infected or not at some given time (not too early, because then the problem is trivial, nor too late, because then the problem is impossible). Shortly after \cite{shah2011rumors}, Pinto et al. proposed a different framework, in which agents (also called \textit{sensors}) reveal, in addition to their state, the time when they became infected, but where only a few of them do so and act as sensors \cite{PintoTV12}; indeed, the problem is trivial if all agents are sensors. This framework is better tailored to epidemics, as it is reasonable that obtaining any information from all the agents is much harder than asking one more question about the starting time of the symptoms of the disease to only some of them. Pinto et al. found that in their framework, if the sensors are already selected, the maximum likelihood estimator of the source has a closed form solution when the underlying network is a tree, and the time it takes for an agent to infect one of its susceptible contacts follows a Gaussian distribution. For general graphs, it is difficult to find an algorithm with any theoretical guarantees, although we note that many heuristics have been developed \cite{hu2018localization,li2019locating,paluch2018fast,paluch2020locating,shen2016locating,tang2018estimating,xu2019identifying,zhu2016locating}. The only exception is on very simple contact networks \cite{lecomte2020noisy}, or when the epidemic spreads deterministically between the agents \cite{zejnilovic2013network}, which is not a realistic assumption for epidemics, but at least the estimation algorithm is trivial, and more emphasis can be put on the question of how the sensors should be selected for good performance \cite{spinelli2017effect,spinelli2018many}, which again is studied by heuristics in the general case \cite{paluch2020optimizing}. For a recent review of source detection algorithms, see \cite{shelke2019source}.

One of the main criticisms of original framework of Pinto et al. is that, even though the contact network is fully known, it is very difficult to find the source exactly unless a large fraction (20-50\%) of the population act as sensors, which is unrealistic in the case of an epidemics, when the source is searched in a large population. An alternative recently proposed is to compute confidence sets for the source instead of finding it \cite{dawkins2021diffusion}. But if our goal is to locate the source exactly, a promising approach is to allow the sensors to be selected adaptively to previous observations \cite{zejnilovic2015sequential,zejnilovic2017sequential}, which we call \textit{adaptive sensor placement}. When the contact network is known, adaptive strategies have been studied by simulations \cite{spinelli2017general,spinelli2017back} and by theoretical analysis \cite{lecomte2020noisy}, and they show a large reduction in the number of required sensors in real networks. In this paper, we will also allow the sensors to be placed adaptively. 

We believe that the most problematic assumption that is still present in source detection papers, is the full knowledge of the contact network of agents, which is unrealistic (let alone because of privacy concerns). Due to this lack of data-availability, algorithms in the source detection literature have not been tested on realistic epidemic data. Moreover, while governmental/health agencies might have access to private datasets, such as cellular location data, from which a contact network may be estimated, these networks may be very noisy, and are potentially unfit for the source detection task. We only know of a few papers that study the effect of imperfections in the network data on the source detection task \cite{mashkaria2020robustness,zejnilovic2016extending}, but these papers study epidemics that spread deterministically between the agents. Inspired by adaptive sensor placement, and by the recent implementations of contact tracing algorithms, we propose a new framework for source detection, which we call Source Detection via Contact Tracing Framework (SDCTF). In SDCTF, algorithms can have two types of queries: contact queries, which can be used to explore the network, and sensor (test) queries, after which agents reveal their symptom onset time as before. The goal of the algorithm is to find the source as accurately as possible, while minimizing the number of contact and sensor queries. The SDCTF is a way to formalize the source detection task; it determines the goal of the algorithm and how information can be gained about the epidemic, but it does not specify the underlying epidemic and mobility data models (simulated or real). In this paper, we analyse different algorithms in the SDCTF with various epidemic and mobility models.

Besides specifying the possible queries that algorithms can make, the SDCTF also determines the way the outbreak is detected, which marks the starting time of the source detection task. In sensor-based source detection, the source detection task often starts long after the outbreak, when essentially all agents in the network are infected \cite{PintoTV12}, which can be seen as a limitation of source detection frameworks. The SDCTF is also closely related to contact tracing frameworks, where it is standard to assign a probability that each node spontaneously self-reports after developing symptoms, which triggers the activation of contact tracing algorithms \cite{kretzschmar2020impact,bradshaw2021bidirectional}. In the SDCTF, we adopt the idea of self-reporting with a slight modification. We believe that the most interesting time to perform the source detection task is when a new disease (or a new mutation of the disease) appears, and therefore we tie these self-reporting events to hospitalizations, where infections are properly diagnosed by healthcare professionals. 
In particular, this means that the SDCTF can only be applied to epidemic data (and models) where hospitalizations are well-defined. In this paper, we use the datasets generated by the Data-driven COVID Simulator (DCS) introduced in \cite{lorch2020quantifying}, which is one of the most realistic toolboxes that generate datasets modelling COVID-19, which we are aware of (notably, hospitalizations are part of the model). We also propose synthetic approximations for the epidemic and mobility models in the DCS; the Deterministically Developing Epidemic model and the Household Network Model, which improve the interpretability of our results since they have fewer parameters.

We propose a simple algorithm called LocalSearch (LS), which adaptively traces back the transmission path from the first hospitalized patient to the source. The LS algorithm is quite efficient at finding the source; the number of contact and sensor queries that it uses does not depend on the size of the network, but only on the local neighborhood of the source. Moreover, the LS algorithm provably finds the source with 100\% accuracy, because of our assumption that every contact and sensor query is answered without noise. However, it is well-known that data-availability is a major issue in contact tracing \cite{beidasrinad2020optimizing}, either because the agents do not comply with contact tracing efforts, or possibly (and in particular in the current COVID-19 epidemic) because they do not develop symptoms, and are unaware that they have the disease. In this paper, we model the effect of asymptomatic agents. When queried and tested, these agents do not reveal their time of infection, only whether they have or had the disease at some point. We show that the accuracy of the LS algorithm drops in the presence of asymptomatic agents, because the algorithm can get stuck while tracing back the transmission path from the first hospitalized patient to the source. Therefore, we propose an improved version of LS called LS+, which accounts for the presence of asymptomatic agents by placing more sensors. We are not aware of any previous work in the source detection literature that models the effect of asymptomatic patients, but the resulting model can be seen as a mix between the snapshot and the sensor-based models. We mention that non-complying agents or agents who provide noisy observations have been studied by \cite{altarelli2014patient,hernando2008fault,louni2015identification}. Non-complying agents could also be included in our framework by treating them as asymptomatic agents (even though in this case we have no information about whether the agent had the disease or not), without jeopardizing the correctness of our algorithms.

We benchmark the LS and LS+ algorithms in both our data-driven and our synthetic epidemic and mobility models, and we compare them to state-of-the-art adaptive \cite{spinelli2017back} and non-adaptive \cite{jiang2016rumor,lokhov2014inferring} algorithms tailored to the SDCTF, whenever possible. We find that both LS and LS+ outperform these baseline algorithms in accuracy (probability of finding the correct source). 

While the LS/LS+ are designed to be simple algorithms, their theoretical analysis is quite challenging. Nevertheless, we are able to provide rigorous results about the success probability of both algorithms after a series of simplifications to the epidemic and mobility models, by extending some recent results on the theory of exponential random trees \cite{feng2018profile,mahmoud2021profile}, which have previously not been connected to the source detection literature. We present these theoretical results in Section~\ref{sec:theory}, after formally introducing the SDCTF, our models and the LS/LS+ algorithms in Section~\ref{sec:models}. By simulations, we show that our analytic results approximate the accuracy of the algorithms well, even in the most realistic setting in Section~\ref{sec:simulation}. Our analytic results provide additional insight into how the parameters of the epidemic and mobility models affect the performance of the algorithms. We discuss these insights along with some non-rigorous computations that mirror our main proof ideas in Section~\ref{sec:boe_sec}. Reading Section~\ref{sec:boe_sec} before Sections~\ref{sec:models}-\ref{sec:simulation} is useful to build intuition, but is not necessary to understand the paper.

\section{Warmup Results}
\label{sec:boe_sec}
\subsection{A Simple Network and Epidemic Model and a Simple Algorithm}
\label{sec:boe_model}
\begin{figure}
\begin{center}
 \includegraphics[width=\textwidth]{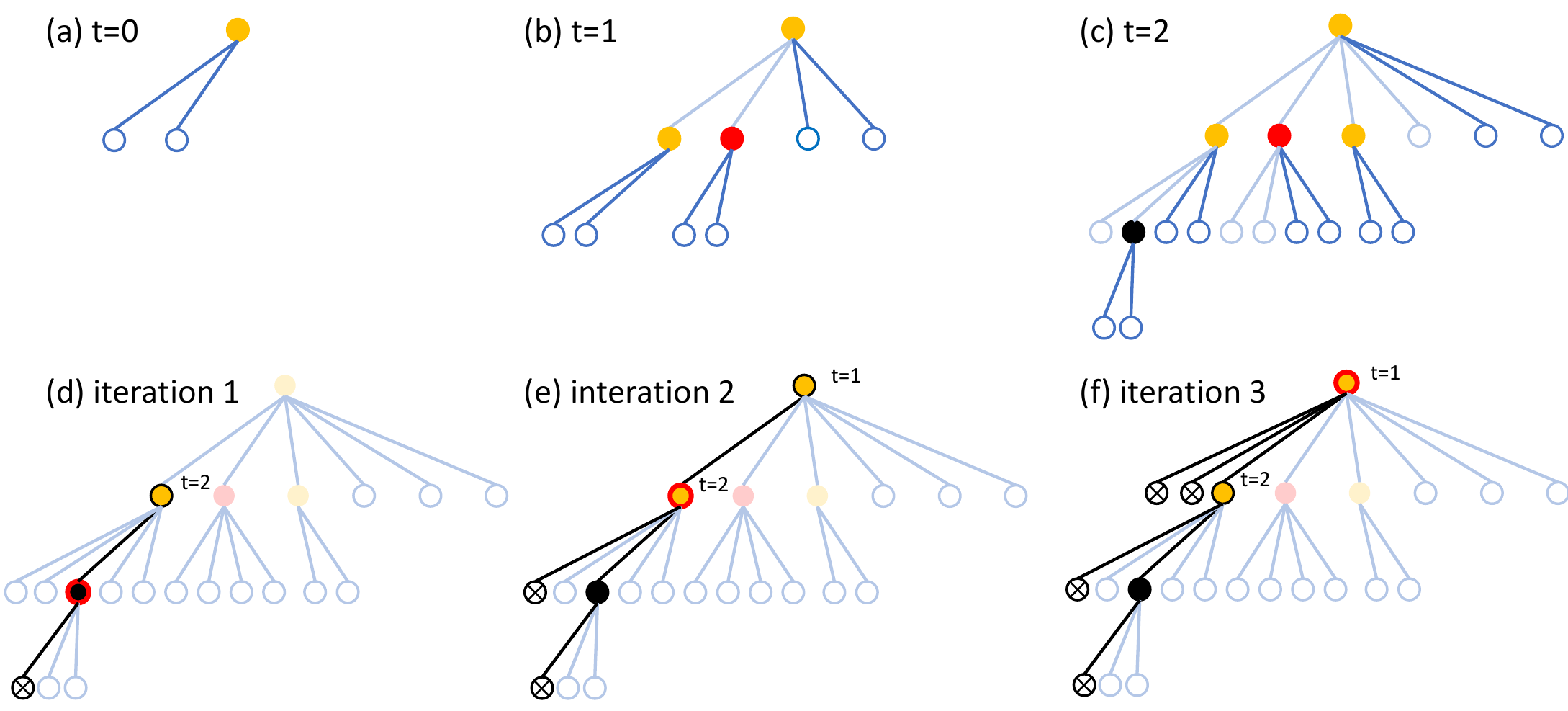}
  \caption{(a)-(c) shows the spread of the infection in the model considered in Section~\ref{sec:boe_model}, which is equivalent to the growth of the RERT, with $d=2$. Dark blue edges show the contacts on day~$t$, and light blue edges show contacts present on previous days (and thus subfigures). Orange (resp., red; black) nodes mark symptomatic non-hospitalized (resp, asymptomatic; symptomatic hospitalized) nodes. (d)-(f) shows the LS source detection algorithm introduced in Section~\ref{sec:boe}, which succeeds in this example because there are no asymptomatic nodes on the transmission path between the first hospitalized node and the source. Black edges show the queried edges, and black stroke marks nodes already discovered by the algorithm. A node with black X marks a negative test result, and red stroked node marks the node currently maintained as source candidate by the LS algorithm.}
  \label{fig:boe}
  \end{center}
\end{figure}

Let us consider a time-dependent network model, where each agent meets $d$ new agents each day in such a way that the contact network is an infinite tree (ignoring the label of the edges giving the propagation time along the edge). This network models homogeneous mixing in a very large population; we consider more realistic network models in Section~\ref{sec:models}. On this network, we consider an epidemic model that starts at $t=0$ with one infected agent, and then progresses as infected agents infect their $d$ susceptible contacts each independently with probability $p_i$ each day. Since our goal is to study the epidemic process, it is sufficient to track only the agents who are already infectious (also called \emph{internal nodes}), and the agents who are in contact with infectious agents at time~$t$ (also called \emph{external nodes}), as shown in Figure \ref{fig:boe} (a)-(c). For $d=1$, the spread of the infection is then equivalent to the growth a random tree $\mathcal{T}_t$ rooted at the source of the infection, known under the name of Random Exponential Recursive Tree (RERT) and recently introduced in \cite{mahmoud2021profile}. Because of the similarities of the models, we refer to the model with general $d$ as RERT in the remaining of this section. We point out that the standard literature on elementary branching processes such as Galton-Watson trees or random recursive trees \cite{drmota2009random} is not applicable in our scenario, because these branching processes have no notion of global time (i.e., a node in such processes becomes infectious immediately after receiving the infection), whereas nodes in diseases commonly go through an exposed, non-infectious period before becoming infectious, which is well captured by the RERT model. We mention that there is literature on more advanced branching processes that do have a notion of global time, e.g. Crump-Mode-Jagers trees \cite{jagers1984growth}, however we opt for the RERT because of its simple definition.

After a node (patient) becomes infected, the disease can take three courses (which for now do not affect $\mathcal{T}_t$): with probability $p_a$ the patient is asymptomatic, with probability $(1-p_a)p_h$ the patient is hospitalized, and with probability $(1-p_a)(1-p_h)$ the patient recovers without hospitalization. The governmental/health agency learns about the outbreak when the first hospitalization occurs (see Figure~\ref{fig:boe} (c)) and starts the source detection process right away. It can inquire about the contacts of each agent and it can test the agents. From patients that were symptomatic (at any point in time in the past), the agency learns about their symptom onset time (which, in this simple model, is always one day after the infection time), but from asymptomatic patients it only learns that they had (or have) the disease at some point when they are tested. The framework introduced in this paragraph (including both the detection of the outbreak through the first hospitalization, and the possible actions the agency can take) is a simplified version of the SDCTF (Source Detection via Contact Tracing Framework), introduced in Section~\ref{sec:SDCTF}.

The network and epidemic models introduced in this section have four parameters: $d, p_i, p_a, p_h$, and it is important to understand how each of them affects the difficulty of source detection in the SDCTF.  We distinguish two important factors. First, if the outbreak is not detected rapidly enough, the length of the transmission path to the first hospitalized agent is long, and source detection becomes then difficult, because a lot of information needs to be recovered. Therefore, a low $p_i$, a low $p_h$ and/or a high $p_a$ parameter can hinder source detection (recall that the probability of hospitalization was $p_h(1-p_a)$). The second factor is related to the difficulty of recovering information about the transmission path. If $p_a$ is high, then there are a lot of nodes who are asymptotic and therefore do not reveal their symptom onset time, making source detection very difficult. Since $p_a$ affects both the length of the transmission path and the amount of collected information, it is safe to expect that, of all parameters, $p_a$ has the largest effect on the difficulty of source detection. The parameter $d$ is interesting, because a large $d$ can reduce the length of the transmission path, but it also makes the information about the transmission path less accessible as more agents need to be tested. Since in this paper we do not set a hard constraint on the total number of available tests, the advantage of a shorter path takes over the drawback of additional tests and a large $d$ increases the success probability. 

To say anything quantitative about source detection in the SDCTF, we must discuss specific algorithms that solve the source detection task. In this paper we propose a simple algorithm called LocalSearch (LS), shown in Figure~\ref{fig:boe}~(d)-(f). The LS algorithm maintains one candidate node $s_c$ at each iteration (initially, the first hospitalized node), which is always symptomatic, and it updates it in a greedy way: at the time of the infection of $s_c$, all its $d$ incident edges are queried, and all its $d$ neighbors are tested. Then the agent with the lowest reported infection time will be the new candidate $s_c$. The algorithm stops when $s_c$ does not change anymore between two consecutive iterations. For simplicity, we assume that the infection does not spread any further during these iterations, however, this assumption does not affect the ability of the algorithm to find the source or not. Indeed, it is not difficult to see that on tree networks, LS succeeds if and only if there are no asymptomatic nodes on the transmission path from the source to the first hospitalized agent. This observation leads us to enhance the LS algorithm by also searching within the neighbors of asymptomatic nodes; we explore this idea in the LS+ algorithm introduced in Section~\ref{sec:localsearch}. We are not aware of this simple greedy algorithm being studied in the context of source detection, although similar ideas were implemented for non-adaptive source detection to lower the runtime of the algorithms \cite{paluch2018fast}.

\subsection{Back of the Envelope Calculation}
\label{sec:boe}

Now, we have all the tools to estimate the probability of success of the LS algorithm. First we condition on the course of the disease in the source. With probability $p_a$, the source is asymptomatic and LS can never succeed. With probability $(1-p_a)p_h$, the source itself becomes hospitalized, and LS always succeeds. Finally, with probability $(1-p_a)(1-p_h)$ the source is symptomatic but not hospitalized, which we call event $\mathcal{A}$. If event $\mathcal{A}$ happens, then LS may or may not succeed depending on whether there are any asymptomatic nodes on the transmission path. More precisely, conditioned on event $\mathcal{A}$ and on the transmission path having length $l$, the probability of success is $(1-p_a)^{l-1}$ (since there are $l-1$ nodes on the path which can be asymptomatic), which implies
\begin{equation}
\label{eq:boe_ps1}
 \P(\mathrm{success}) = (1-p_a)p_h+ (1-p_a)(1-p_h) \left(\sum_{l=1}^t \P\left(\text{transmission path has length $l$} \mid \mathcal{A} \right)
    (1-p_a)^{l-1} \right).
\end{equation}

The difficult part is to compute the distribution of the transmission path conditioned on event $\mathcal{A}$; indeed we already saw that all four parameters $d, p_i, p_a, p_h$ affect this distribution in a non-trivial way. Let us perform a back of the envelope computation to get more insight into the effect of these parameters. The exact structure of the infection tree will not matter for this computation, only its \textit{profile} does. It is denoted by $\mathcal{T}_t(l)$ and defined as the number of (internal) nodes at level $l$ (i.e., at distance $l$ from the source of the infection). Remember that by definition the RERT has $d \cdot \mathcal{T}_{t-1}(l-1)$ external nodes on level $l$, and that at time~$t$ each external node is promoted to be internal with probability $p_i$ to form $\mathcal{T}_t$. Consequently, the level of a node~$h$ added at time $t>0$ has the same distribution (conditioned on the tree $\mathcal{T}_{t-1}$ at the previous step) as the size (number of internal nodes) of the profile $\mathcal{T}_{t-1}(l-1)$, that is, \begin{equation}
\label{eq:back_of_envelope_def}
    \P(\mathrm{level}(h)=l \mid \mathcal{T}_{t-1}=T_{t-1}) = \frac{T_{t-1}(l-1)}{|T_{t-1}|}.
\end{equation}
Working on the RERT directly can be a daunting task, therefore we propose to approximate the numerator and the denominator of equation \eqref{eq:back_of_envelope_def} by $\E[\mathcal{T}_{t-1}(l-1)]$ and $\E|[\mathcal{T}_{t-1}|]$, respectively. It can be shown by a simple inductive argument, or by generating functions as in \cite{mahmoud2021profile}, that for RERTs we have $\E[\mathcal{T}_{t}(l)]=\binom{t}{l} (dp_i)^l$ and $\E[|\mathcal{T}_{t}|]=(1+dp_i)^t$, which suggests a binomial distribution for the level of $h$. And indeed, we can approximate the distribution of the level of a node $h$ added at time $t$ as
\begin{align*}
    \P(\mathrm{level}(h)=l) &\approx \frac{\E[\mathcal{T}_{t-1}(l-1)]}{\E[|\mathcal{T}_{t-1}|]} \\
    &=\frac{\binom{t-1}{l-1} (dp_i)^{l-1}}{(1+dp_i)^{t-1}}\\
    &=\binom{t-1}{l-1}\left(\frac{ dp_i}{1+dp_i}\right)^{l-1} \left(1-\frac{ dp_i}{1+dp_i}\right)^{t-l} \\
    &= \P(\mathrm{Bin}(t-1,q) = l-1),
\end{align*}
with $q=dp_i/(1+dp_i)$.

One of the main challenges of this calculation is that we do not know the day of the first hospitalization $t$ conditioned on event $\mathcal{A}$, we only know that each node is hospitalized with probability $(1-p_a)p_h$, which means that the index of the first hospitalized node follows a geometric distribution with mean $1/((1-p_a)p_h)$. We approximate $t-1$ by the first time that the expected size of the infection tree (excluding the source since we condition on event $\mathcal{A}$) exceeds the expected index of the first hospitalized node. Therefore we solve
$$\E[|\mathcal{T}_{t-1}|-1]=(1+dp_i)^{t-1}-1 = \frac{1}{(1-p_a)p_h}=\E[\text{index of the first hospitalized node}]$$
for $t$ (relaxing the constraint that $t$ is an integer), which gives 
$$t-1= \frac{\log \left(1+\frac{1}{(1-p_a)p_h} \right)}{\log(1+dp_i)}.$$ 
Consequently, we approximate $\P\left(\text{transmission path has length $l$} \mid \mathcal{A} \right)$ by $\P(\mathrm{Bin}(t-1,q) = l-1)$. Continuing equation \eqref{eq:boe_ps1}, and using the well-known expression of the probability generating function of the binomial distribution, we get
\begin{align}
\label{eq:two}
    \P(\mathrm{success}) &\approx (1-p_a)p_h+ (1-p_a)(1-p_h) \left(\sum_{l=1}^t \P\left(\mathrm{Bin} \left(t-1, q \right)=l-1 \right) (1-p_a)^{l-1} \right) \nonumber \\
    &= (1-p_a)\left(p_h + (1-p_h) \left((1-p_a)\frac{dp_i}{1+dp_i}+1-\frac{dp_i}{1+dp_i}\right) ^{\frac{\log\left(1+\frac{1}{(1-p_a)p_h}\right)}{\log(1+dp_i)}} \right).
\end{align}

One can check that this expression agrees with our qualitative intuition. However, it is not at all clear whether it is valid because of the strong approximations made in some steps of the above computation. In Section \ref{sec:theory}, we prove a rigorous upper bound on the success probability, and we also provide much more careful approximations by proving exact theorems about the simplified models that we use. Then, in Section \ref{sec:simulation} we compare our results with simulation results on synthetic data, as well as with data generated by the DCS model.

\section{Models, Methods, Algorithms}
\label{sec:models}

\subsection{Epidemic Models}
\label{sec:epidemic_models}
\subsubsection{The DCS Model}
\label{subsub:DCS}

We call DCS the model implemented by \cite{lorch2020quantifying}. The DCS model is fairly complex, and we only give a brief overview.

Each agent in the agent set $V$ can be in one of 8 states: susceptible, exposed, asymptomatic infectious, pre-symptomatic infectious, symptomatic infectious, hospitalized, recovered or dead. Transitions between different states are characterized by counting processes described by stochastic differential equations with jumps. The most important, and also most complicated of these counting processes is the exposure counting process $N_i(t)$, which is modeled by a Hawkes process for each agent~$i$. Hawkes processes are point processes with a time-dependent, self-exciting conditional intensity function $\lambda^*_i(t)$.
\begin{equation}
\label{eq:Hawkes}
\hspace*{-5pt}
\lambda^*_i(t) = 
\beta
\sum_{j \in V \backslash \{i\}} 
    \int_{t - \delta}^{t} K_{i,j}(\tau) ~ \gamma e^{-\gamma(t-\tau)} \, d\tau
\end{equation}
where the kernel $K_{i,j}(\tau)$ indicates whether $j$ has been at time $\tau$ at the same site where $i$ is at time $t$, and whether $j$ is in the infectious state. Parameters $\gamma$ and $\delta$ are the decay of infectiousness at sites and the non-contact contamination window, respectively, and they account for the fact that $j$ can infect $i$ even if they are never at the same site, as $j$ can leave some pathogens behind (airborne for instance). Parameter $\beta$ is the transmission rate for symptomatic and asymptomatic individuals, and it comes in two versions: $\beta_c$ accounts for infections outside the household and $\beta_h$ accounts for infection in the household. Parameters $\beta_c$ and $\beta_h$ are fitted to the COVID-19 infection data of Tubingen from 12/03/2020 to 03/05/2020 using Bayesian Optimization. The model also has a parameter for the relative asymptomatic transmission rate built into the function $K_{i,j}(\tau)$, which scales down the infectiousness of asymptomatic agents (to 55\% of the infectiousness of symptomatic agents by default). 

Once a susceptible agent becomes infected, the disease can take three possible courses (see Figure~\ref{fig:mobility_models}~(a)). With probability $p_a$, the agent becomes asymptomatic infectious after time $T_E$, and then recovers after time $T_I$. With probability $1-p_a$, the agent becomes pre-symptomatic infectious after time $T_E$, next symptomatic infectious after time $T_P$, and then recovers with probability $1-p_h$ after time $T_I-T_P$, or becomes hospitalized with probability $p_h$ after time $T_H$. Agents in the DCS are also assigned age values based on demographic data, and the hospitalization probability $p_h$ of each agent is determined based on its age (following COVID-19 infection data). The times $T_E, T_P, T_I$ and $T_H$ are drawn from an appropriately parametrized (using values from the COVID-19 literature) lognormal distribution as shown in Table \ref{tab:Tvalues}.
 
\subsubsection{The DDE Model}
\label{subsub:DDE}

We start by taking the DCS model \cite{lorch2020quantifying}, which we simplify to enable its theoretical analysis. In the Deterministically Developing Epidemic (DDE) model, continuous time (used in DCS) is replaced by discrete time-steps: we refer to one time-step in the DDE as one day. Instead of modelling the infection propagation as a Hawkes process, an infectious agent (symptomatic or asymptomatic) can infect its susceptible neighbor with probability $p_i$ each day. Thereafter, the disease progresses the same way as in the DCS, except that in the DDE model the transition times are deterministic (the infection events and the severity of the disease (i.e., the (a)symptomatic and hospitalized states) are still determined randomly), and we have a single parameter $p_h$ for the hospitalization probability (agents in this model do not have an age parameter). We discuss how we set the parameters of the DDE model in Section \ref{sec:parameters}.

\begin{figure}
\begin{center}
 \includegraphics[width=\textwidth]{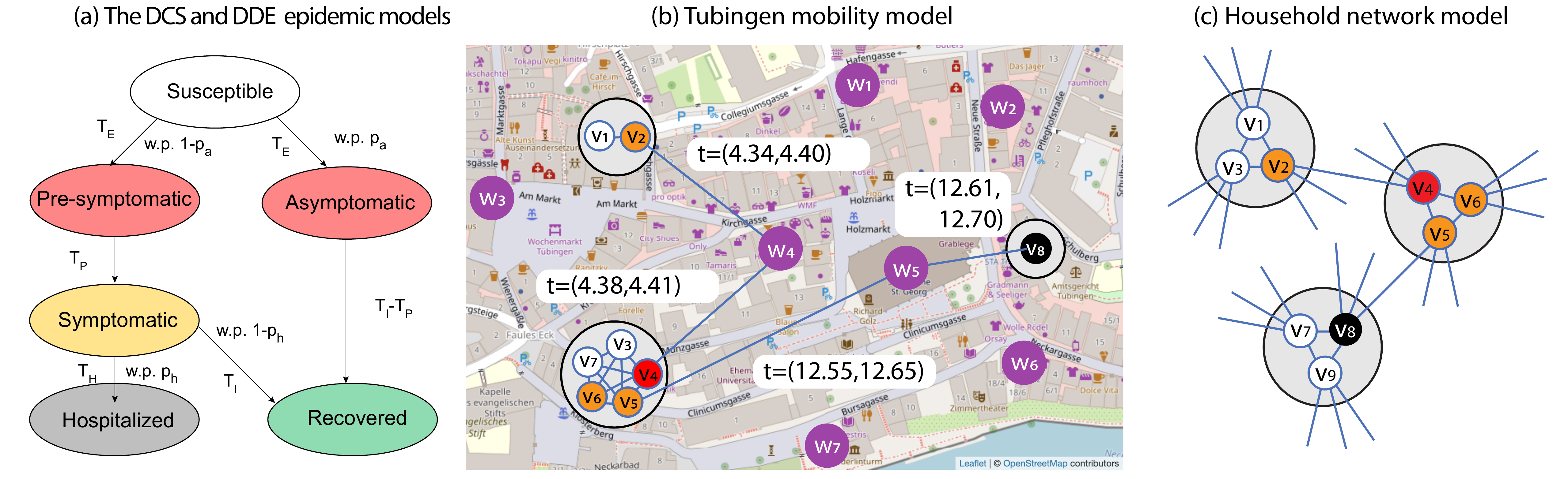}
  \caption{(a) The flow diagram of the DCS and DDE epidemic models. (b) A possible epidemic outbreak in the Tubingen mobility model, and (c) the Household network model. The large grey circles mark households, and the purple nodes mark places, otherwise we use the same coloring as in (a). In both cases (b) and (c), the transmission paths are $(v_2, v_4, v_5, v_8)$. }
  \label{fig:mobility_models}
  \end{center}
\end{figure}

\subsection{Simulating Mobility}
\label{sec:mobility}

\subsubsection{Tubingen Mobility Model}
\label{sec:TU}

We briefly review the mobility model introduced in \cite{lorch2020quantifying}, and illustrated in Figure~\ref{fig:mobility_models}~(b). The population is partitioned into households of possibly varying size (usually between 1 and 5). The households are assigned a location, and we also place some external sites (shops, offices, schools, transport stations, recreating sites) on the map, which the agents may visit.  The location of the households and the number of agents in them is sampled randomly based on demographic datasets. Initially, each agent is assigned a few favorite sites (randomly based on distance), and will only visit these throughout the simulation. Each agent decides to leave home after some exponentially distributed time, visits one of its (randomly chosen) favorite sites, and comes back home after another (usually much shorter) exponentially distributed time. If two agents visit the same site at the same time, or within some time $\delta$, we record them as a contact, which gives an opportunity for the infection to propagate. We denote the Tubingen mobility model as TU, and the DCS epidemic model that runs on the TU mobility model as DCS+TU. 

\subsubsection{Household Network Model}
\label{sec:HNM}

The Household network model (HNM) was inspired by \cite{lorch2020quantifying}, however we note that similar models have been studied in the theoretical community by~\cite{ball2009threshold}. As in the Tubingen mobility model, in HNM $N$ nodes are assigned into households, but of constant size $d_h+1$. Every pair of nodes in the same household are connected by an edge, forming therefore cliques of size $d_h+1$. Additionally, each node is assigned $d_c$ half edges, which are paired uniformly at random with other half-edges in the beginning. Some half-edge pairings can result in self-loops or multi-edges, which are discarded. This construction defines a random graph generated by a configuration model, which shares a lot of similarities with Random Regular Graphs (RRG) \cite{wormald1999models}. In fact, if we join nodes in the same household into a single node in the HNM (which we refer to as the \textit{network of households} of the HNM), then the resulting graph is equivalent to the \textit{pairing model} of RRGs with degree $d_c(d_h+1)$. It is well-known that in the pairing model of RGGs of degree $d$, the local neighborhood (of constant radius, as the number of nodes tends to infinity) of a uniformly randomly chosen vertex is a $d$-regular tree (with probability tending to 1), which implies that locally there are asymptotically almost surely no self-loops, multi-edges or any cycles in the graph. This result has various names; in random graph theory the result is usually proved by subgraph counting \cite{wormald1999models}, in probability theory it is the basis of branching process approximations \cite{ball2009threshold}, and in graph limit theory it is called the local convergence to the infinite $d$-regular tree~\cite{benjamini2011recurrence}. In our theoretical analysis, this result motivates the approximation of the neighborhood of the source in the network of households of the HNM by an infinite $d_c(d_h+1)$-regular tree. The HNM itself is then approximated by replacing each (household) node of the infinite $d_c(d_h+1)$-regular tree of households by a $(d_h+1)$-clique, and by setting the edges so that each (individual) node has degree exactly $d_c+d_h$, while keeping the connection between cliques unchanged (see Figure~\ref{fig:mobility_models}~(c) for a visualization). 

Since the HNM is a time-independent graph, we adopt the standard notations from graph theory. Formally, the HNM is given by the set of nodes and edges $G=(V,E)$. Let us denote by $H(v)$ the set of nodes that are in the same household as node $v$. The distance between two nodes $u,v \in V$ (denoted by $d(u,v)$) is defined as a number of edges of the shortest path between $u$ and $v$. We denote the DDE epidemic model that runs on the HNM network as DDE+HNM. 

\subsection{The Source Detection via Contact Tracing Framework}
\label{sec:SDCTF}

We present the Source Detection via Contact Tracing Framework (SDCTF), which can be applied to both epidemic and mobility models presented so far. The framework determines how the government/health agency, which conducts the source detection task, learns about the outbreak, and how it can gather further information to locate the source. In the SDCTF, as in Section \ref{sec:boe_model}, the agency learns about the outbreak when the first hospitalization occurs, and it also learns the identity of nodes when they become hospitalized (including the identity of the first hospitalized node).

After the outbreak is detected, the agency can make three types of queries. The first type of query, the household query with parameter $v$, reveals the agents that live in the same household as $v$. The household query works the same way in both the TU and the HNM models, and we do not limit the number of times it can be called (these queries are considered as cheap in the SDCTF). The second type of query, the contact query, works differently in the TU and the HNM models. For the TU model, a contact query has two parameters: an agent $v$ and a time window $[t_1,t_2]$. As a result, all agents that have been in contact with $v$ (and therefore could have infected $v$ or could have been infected by $v$) at an external site between $t_1$ and $t_2$ are revealed. In the HNM, no time window is needed for the contact query (which we also call edge query), and all neighbors of $v$ in graph $G$ are revealed. Contact (and edge) queries are considered expensive in the SDCTF. While in this paper we do not limit the number of available queries, we track the number of contacts and edges that are revealed as the algorithm runs. Note that in the TU model if two agents $v_1$ and $v_2$ have been in contact during the time window $[t_1, t_2]$ and also during a different time window $[t_3,t_4]$, then those are counted as separate contacts, whereas in the HNM an edge between $v_1$ and $v_2$ is only counted once. Although contact queries are considered expensive, both household and contact queries are answered instantly in the SDCTF.

The third kind of query is the test query with parameter $v$, which reveals information about the course of the disease in the queried agent (see Figure \ref{fig:mobility_models}~(a)). Symptomatic patients reveal the time of their symptom onset (which exactly determines their time of infection in the DDE due to the deterministic transition times) if they are past the pre-symptomatic state (i.e., if they are either infectious or recovered). Asymptomatic and pre-symptomatic patients do not reveal any information about their infection time; they just reveal that they have the disease or had the disease at some point and have recovered. For all algorithms we assume that asymptomatic patients do not reveal whether they have the infection at the time they are queried. Finally, agents who have not been exposed, or are still in their exposed state, give a negative test result. Test queries are again considered expensive in the SDCTF, we even limit the population that can be tested on any given day to at most 1\% of the total population, due to the capacity of testing facilities. However, since in this paper we do not limit the number of days that the algorithm can use to locate the source, the limit on the number of tests does not play an important role. As opposed to household and contact queries (and the model in Section \ref{sec:boe_model}), tests results are only answered the next day in the SDCTF, which means that the algorithms must operate in ``real-time'', while the epidemic keeps propagating.

\subsection{Parameters}
\label{sec:parameters}

\begin{figure}[h]
\begin{center}
 \includegraphics[width=\textwidth]{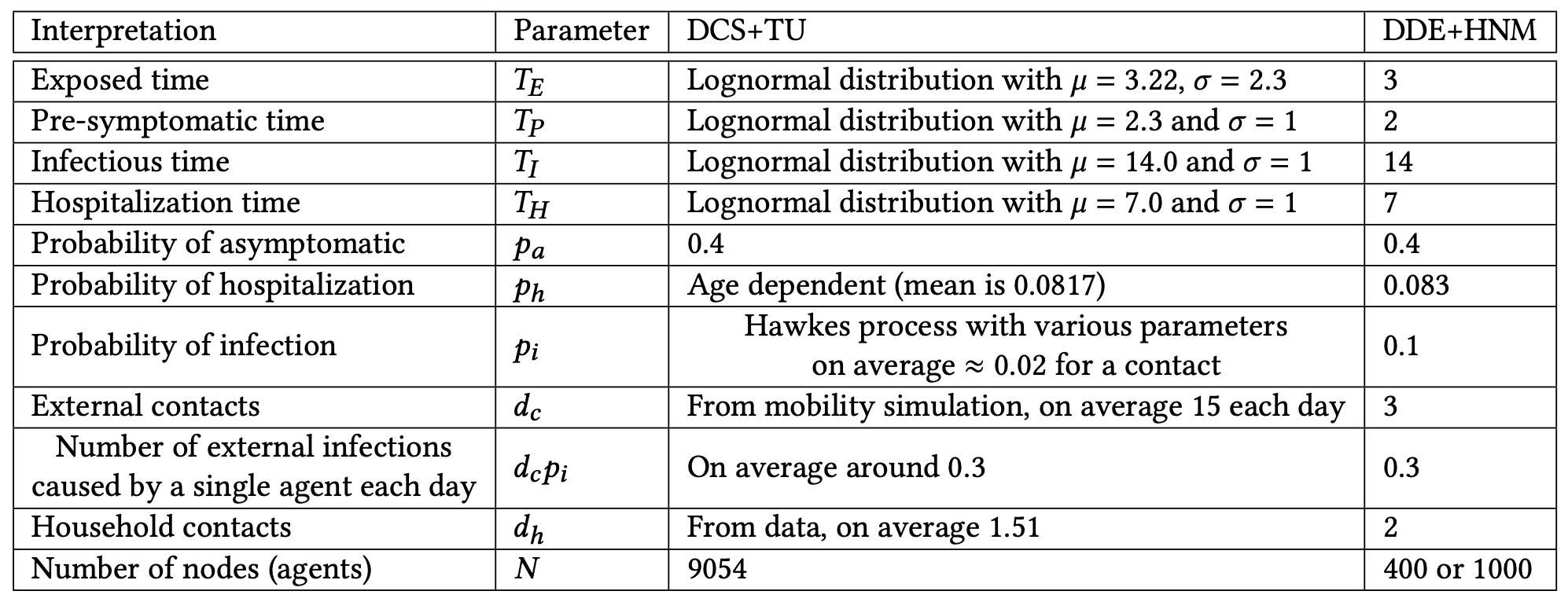}
  \caption{Default values for the infection parameters in the DCS+TU and the DDE+HNM models.}
  \label{tab:Tvalues}
  \end{center}
\end{figure}

The DCS+TU model has many parameters, most of which are fitted to COVID-19 datasets of Tubingen from 12/03/2020 to 03/05/2020 by \cite{lorch2020quantifying} (we show the most relevant parameters in Table \ref{tab:Tvalues}). We determined the parameters of the DDE+HNM model so that they fit the parameters of the DCS+TU as closely as possible (see the precise values in Table \ref{tab:Tvalues}). We determine the values of $T_E, T_P, T_I$ in the DDE+HNM by rounding the expected value of the corresponding distribution in the DCS+TU to the nearest integer. Since $p_a$ is simply a constant in both models, we keep the same numerical value in the DDE+HNM. The parameter $p_h$ is more complicated, because in the DCS+TU model there is a different hospitalization probability for each age group. We take the average hospitalization probability across the population to be $p_h$. The most complicated parameter to fit is $p_i$, because in the DCS+TU model, infections are modelled by a Hawkes process, which depends on many parameters, including whether the infectious agent is symptomatic or asymptomatic, the length of the visit, the site where the infection happens, etc (see equation \eqref{eq:Hawkes}). We empirically observe the probability of infection in every contact in several simulations, and we find that an agent has on average 15 contacts outside the household each day, and that the average probability of infection during such a contact is around 0.02. However, since we use smaller networks for the DDE+HNM ($N=400$ or $1000$, because running the baselines on larger networks is not feasible) than the DCS ($N=9054$), setting $d_c$ to be as high as 15 would violate the assumption that the network of households of the HNM can be locally approximated by a tree (see Section~\ref{sec:HNM}). Therefore we chose $d_c=3$ for the HNM and we scale $p_i$ so that $d_cp_i$ (the expected number of external infections caused by a single agent each day) is the same in the DCS+TU and the DDE+HNM models. Finally, we choose $d_h$ in the DDE+HNM by rounding the average household connections in the DCS+TU. Note that the average number of household connections is not the same as the average number of household members, because the number of connections grows quadratically in the size of the households, and thus fitting to the number of connections results in a higher $d_c$ (due to the Quadratic Mean-Arithmetic Mean inequality).

Finding the default values for the parameters is useful to create a realistic model. However, we also interested in the effect of each of the parameters on the performance of our algorithms. Therefore, in the DDE+HNM, we vary the parameters $p_a,p_h,p_i,d_h$ and $d_c$, while keeping the other ones unchanged. For the DCS+TU model, we also keep the mobility model fixed and we focus on varying the parameters $p_a,p_h$ and $p_i$. As noted above, there is no single parameter $p_h$ or $p_i$ in the DCS+TU model, therefore we change all hospitalization probabilities and all intensities of the Hawkes processes so that the hospitalization probability averaged across the population and the infection probability averaged across contacts equal the desired values.

\subsection{The LocalSearch Algorithms LS and LS+}
\label{sec:localsearch}

The LS algorithm finds patient zero by local greedy search.
It keeps track of a candidate node, which is always the node with the earliest reported symptom onset time. We denote the candidate of the algorithm at iteration $i>0$ by $s_{c,i}$. We think of $s_c$ as a list, which is updated in each iteration of the algorithm, and we use the notation $s_{c,-1}$ for the last element of the list (i.e., the current candidate). In each iteration of the algorithm, we compute a new candidate denoted by $s_c'$, and we append it at the end of the list $s_c$ at the beginning of the next iteration, unless $s_c'=s_{c,-1}$, in which case the algorithm terminates. 

\begin{figure}[h]
\begin{center}
 \includegraphics[width=\textwidth]{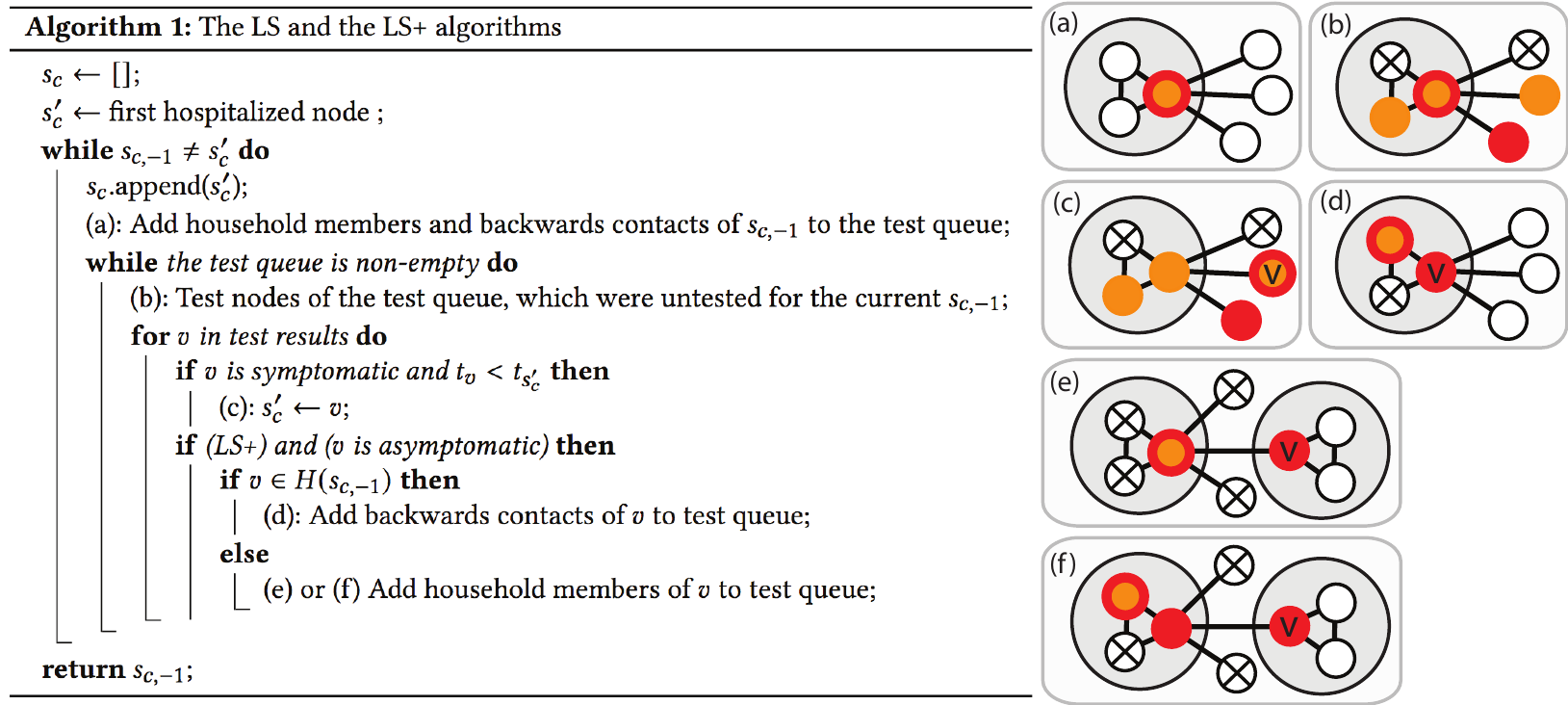}
  \caption{Pseudocode and graphical explanation for the LS and LS+ algorithms. We use the same coloring as in Figure \ref{fig:mobility_models}~(a). Black edges show the queried edges, a node with black X marks a negative test result, and red stroked node marks the node currently maintained as source candidate by the LS algorithm. We denote by $t_v$ the symptom onset time of symptomatic node $v$ and by $H(v)$ the household of a node $v$ similarly to the main text. }
  \label{fig:Alg1_figure}
  \end{center}
\end{figure}

Since we consider the SDCTF, the outbreak is detected when the first hospitalized case is reported. At that time, $s_c'$ is initialized to be the hospitalized patient, the test queue is initialized to be empty, and the algorithm is started. In the beginning of an iteration, if the test queue is empty, the household members and the ``backward'' contacts of the current candidate $s_{c,-1}$ are queried and are added to the test queue (see Figure~\ref{fig:Alg1_figure}~(a)). We define ``backward'' contacts as the set of nodes that have been in contact with $s_{c,-1}$ in the interval $[t_{s_{c,-1}} - (T_E+T_P) - (\sigma_E +\sigma_P),t_{s_{c,-1}} - (T_E+T_P) +  (\sigma_E +\sigma_P)]$, where $t_{s_{c,-1}}$ is the symptom onset time of current candidate $s_{c,-1}$. The terms $\sigma_E$ and $\sigma_P$ model the standard deviation of the transition times, and they are set to zero for the DDE and to $\sigma_E=2$ and $\sigma_P=1$ for the DCS based on Table \ref{tab:Tvalues}. We note that the notion of ``backward'' contacts is only meaningful in the case of time-dependent network models; for the HNM, all neighbors are counted as backward contacts.
 
After the test queue is initialized, the agents inside the queue are tested (see Figure \ref{fig:Alg1_figure}~(b)). Not all nodes can be tested on the same day because of the limitation on the number of tests available per day in the SDCTF, however, this has little effect because we do not proceed to the next iteration until the test queue becomes empty. Once the test results come back to the agency, if any of the (symptomatic) nodes $v$ reports an earlier symptom onset time than the current candidate $s_{c,-1}$, then we update our next candidate $s_c'$ to be $v$ (see Figure \ref{fig:Alg1_figure}~(c)). We note that the iteration does not stop immediately after $s_c'$ is first updated; the iteration runs until the test queue becomes empty, and until then, $s_c'$ can be updated multiple times. This is important in the theoretical results to prevent the algorithm from getting sidetracked (see Figure~\ref{fig:LSp_cases}). We also experimented with a version of the LS and LS+ algorithms where the iteration stops immediately once $s_c'$ is updated; we call these algorithms LSv2 and LS+v2.

The main drawback of the LS algorithm is that is gets stuck very easily if there is even one asymptomatic node on the transmission path. For this reason, we introduce the LS+ algorithm, in which we enter the backward contacts of the asymptomatic household members of $s_{c,-1}$, and the household members of any asymptomatic node into the testing queue (see Figure \ref{fig:Alg1_figure}~(d)-(f)). Since the symptom onset times of asymptomatic nodes $v$ are not revealed, we define backward contact in this case as any contact in the time window $[t_{s_{c,-1}} -(T_P+2T_E+T_I), t_{s_{c,-1}} - (T_P+2T_E)]$, where $t_{s_{c,-1}}$ is still the symptom onset time of the current candidate $s_{c,-1}$. Indeed, in the DDE model, since $s_{c,-1}$ was infected at $t_{s_{c,-1}} -(T_P+T_E)$, if $v$ infected $s_{c,-1}$, agent $v$ must have been infectious at that time, which implies that $v$ could not have been infected later than $t_{s_{c,-1}} -(T_P+2T_E)$ or earlier than $t_{s_{c,-1}} -(T_P+2T_E+T_I)$. In the DCS model, the terms $\sigma_E$ and $\sigma_P$ can be subtracted and added to the two ends of the queried time window to account for the randomness in the transition times.

Both algorithms stop if the testing queue becomes empty before a node with an earlier symptom onset time than $s_{c,-1}$ is discovered, and both algorithm return $s_{c,-1}$ as their inferred source. The high level pseudocode and an illustration of the LS and LS+ algorithms are given in Figure~\ref{fig:Alg1_figure}.

%

\section{Theoretical Results}
\label{sec:theory}

In this section we present theoretical results for the LS and LS+ algorithms described in Section~\ref{sec:localsearch}. We follow a similar approach as in the non-rigorous computation in Section~\ref{sec:boe}, which useful but not necessary for understanding this section. All the statements are rigorously established, and whenever we reach a point where the computations would become intractable, we propose a simpler approximate model to study. One of the main contributions of this paper is to identify which computations can be done on more general models, and which computations need more simplified ones (see Figure~\ref{fig:overview} for an overview of the different models used for the computations in this section).

We compute the success probability of the LS and LS+ algorithms in two steps. We first assume the length of the transmission path known in Section~\ref{sec:lsp_succes}
. This computation  is then made possible by a tree approximation of the HNM, called the Red-Blue (RB) tree (defined in Section~\ref{subsub:RBtree}), and a slightly modified version of the DDE model called $\mathrm{DDE}_{\mathrm{NR}}$ (defined in Section~\ref{subsub:DDENR}).  The RB tree preserves some of the household structure in the HNM, and therefore allows us gain insight into the difference between the LS and LS+ algorithms, which would be difficult to obtain if we had worked on trees without taking the household structure in account. 

For the second step, we would need to compute the distribution of the transmission path on the RB tree. However, finding a closed form expression is intractable. Instead, we combine the network and epidemic models into a growing random tree model, and we consider a $d$-ary Random Exponential Tree (RET).
The $d$-ary RET model has only been studied for $d=2$ \cite{feng2018profile}; we extend the results on their expected profile for general $d$ in Section~\ref{subsub:RET}. Nevertheless, working on $d$-ary RETs still remains difficult, and therefore, in our last modeling step, we introduce a Deterministic Exponential Tree (DET) model, whose profile is close to the expected profile of the RET, and we compute the distribution of the transmission path on this model in Section~\ref{subsub:DET}.

\begin{figure}
\begin{center}
     \begin{subfigure}[b]{\textwidth}
        \caption{}
         \includegraphics[width=\textwidth]{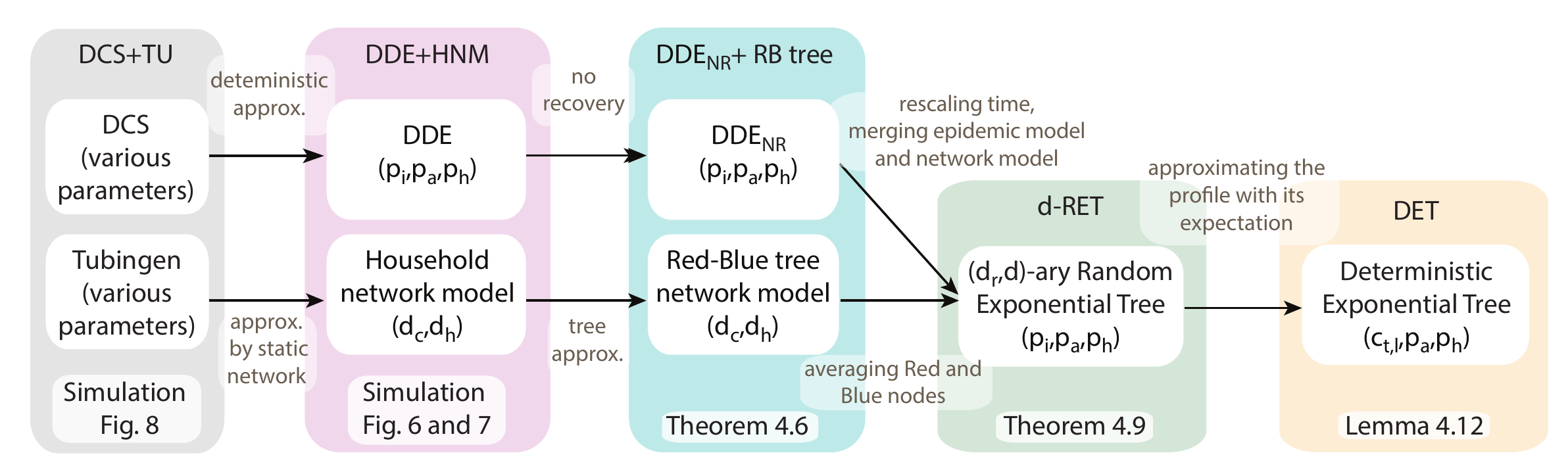}
         \label{fig:y equals x}
     \end{subfigure}
     \hfill
     \begin{subfigure}[b]{\textwidth}
        \caption{}
         \centering
         \includegraphics[width=\textwidth]{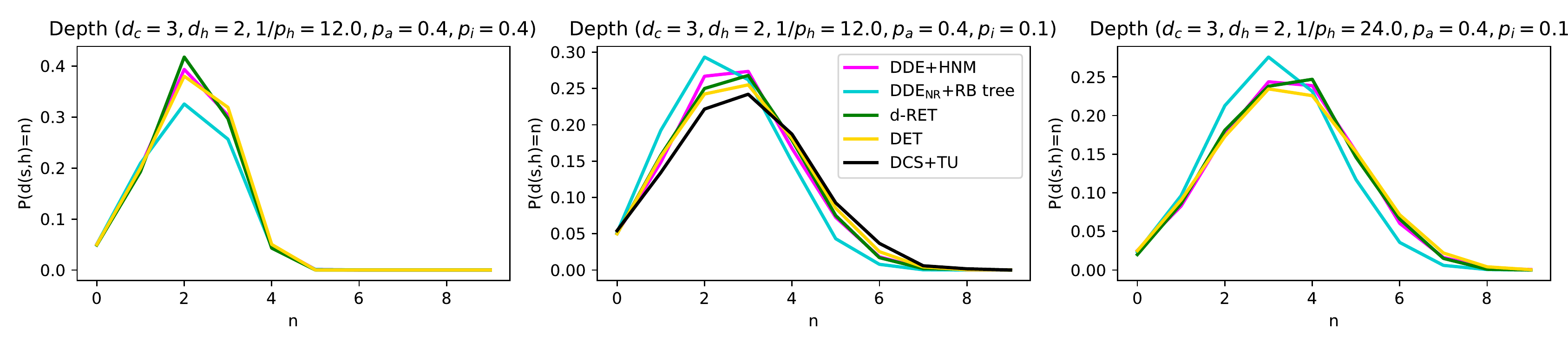}
         \label{fig:three sin x}
     \end{subfigure}
  \caption{The different approximation methods (a) and the distribution of the length of transmission path in the different models (b) proposed in Section~\ref{sec:theory}. Panel (b) also shows the length of the transmission path in the DCS model on the TU dynamics, to highlight the fit of our model.}
  \label{fig:overview}
\end{center}
\end{figure}

To summarize all models considered in this paper, we have a data-driven and a synthetic model for simulations (DCS+TU and HNM+DDE), an analytically tractable model (RB-tree+$\mathrm{DDE}_{\mathrm{NR}}$) where we can compute the success probability if the length of the transmission path is known. In a second stage, we compute the distribution of a transmission path on a deterministic tree (DET), which has a similar profile as a random tree (RET) that approximates our analytically tractable model. We visualize these five different models in Figure~\ref{fig:overview}~(a), and we show by simulations in Figure~\ref{fig:overview}~(b) that the distribution of the transmission path is similar in all of the considered models with appropriately scaled parameters. We compare our analytic results on the success probabilities of the LS and LS+ algorithms with our simulation results in Section \ref{sec:comp_theory} in Figure~\ref{fig:theory_plots}.


\subsection{Success Probability of LS and LS+ Algorithms on the RB Tree}
\label{sec:lsp_succes}
In this section we introduce the Red-Blue (RB) tree model (which is a tree approximation to the HNM), and we calculate the exact probability that the LS and LS+ algorithms succeed, if the length of the transmission path is known. 

\subsubsection{Red-Blue tree models}
\label{subsub:RBtree}
In short, a RB tree is a two-type branching process with a deterministic offspring distribution that depends on $d_h$ and $d_c$. The lack of randomness in this distribution makes us adopt the formalism of deterministic rooted trees.
\begin{definition}
\label{def:RBtree}
Let a rooted tree, denoted by $G(s)$, be a tree graph with a distinguished node root node $s$. Let $u$ and $v$ be two nodes connected by an edge in $G(s)$. If $d(u,s) < d(v,s)$, we say that $u$ is a parent of $v$, otherwise $u$ is a child of $v$. Moreover, if $d(s,v) = l$ we say that $v$ is on level $l$. An RB tree with parameters $(d_c, d_h)$ is an infinite rooted tree, such that the nodes also have an additional color property. The root is always colored red and the rest of the nodes are colored red or blue. The root has $d_c$ red and $d_h$ blue children. Every other red node has $d_c-1$ red and $d_h$ blue children, and every blue node has $d_c$ red children and no blue children. Red nodes and their $d_h$ blue children partition the nodes of the RB tree $G(s)$ into subsets of size $d_h+1$, which we call households.
\end{definition}
\begin{remark}
\label{rem:RBdef}
In the RB tree, each blue node has degree $d_c+1$, and each red node has degree $d_c+d_h$, including the root of the tree $s$ (which is 
the source of the epidemic, when the RB tree is combined with an epidemic model).
\end{remark}
The RB tree can be seen as a local tree approximation of the HNM. Let $G=(V,E)$ be an HNM with parameters $(d_c,d_h)$, and let $s \in V$ be the distinguished source node. In Section \ref{sec:HNM} we noted that the HNM can be approximated locally around the source node by replacing each node of an infinite $d_c(d_h+1)$-regular tree by a $(d_h+1)$-clique, and setting the edges so that each node has degree exactly $d_c+d_h$, while keeping the connection between cliques unchanged. Let us call this infinite graph $G^*$. Although $G^*$ is not a tree, all cycles in $G^*$ must be contained entirely inside the households, which implies that in each household there exists exactly one node that has the minimal distance to the source. We will refer to these nodes with minimal distance to the source as the red nodes, and we color the rest of the nodes blue. In other words, the red nodes will be the first ones in their households to be infected. Let us now delete the edges between the blue nodes in $G^*$ to obtain graph $G'$. We claim that $G'$ is isomorphic to the RB tree $G(s)$ rooted at the source~$s$. Indeed, since the edges between blue nodes have been deleted in $G^*$ to form $G'$, each blue node has $d_c+1$ red neighbors and no blue neighbor, and since the edges incident to red nodes have been unchanged, each red node has $d_c$ red and $d_h$ blue neighbors, exactly as in the definition of RB tree above.

Note that a household in $G^*$ is completely characterized by only specifying the colors of the nodes: a household always consists of one red node and of its $d_h$ blue children. We use this characterization as a definition for households in the RB tree $G'$, because it does not depend on the edges from $G$ that are deleted in $G^*$, whereas this deletion makes the original definition of a household as a clique in $G$ unusable.

Next, we make some important observations the behaviour of the LS and the LS+ algorithms on RB trees, which we prove in Appendix \ref{sec:LS,LSp,trees_app}. We start by formalizing the notion of transmission path. 

\begin{definition}
\label{def:inf_path}
Let $h$ be the first hospitalized node and $s$ be the source. We call the path $(s = v_0, v_1, ... v_l = h)$, where $v_{i}$ is the infector of $v_{i+1}$ for $0 \leq i < l$, the \emph{transmission path}. Also we call the path $(v_l, v_{l-1}, ... v_1)$ the \emph{reverse transmission path}.
\end{definition}

\begin{remark}
\label{rem:RBtree}
Note that in an RB tree, each household traversed by a transmission path shares one (the red node in the household) or two (the red node of the household and one of its $d_h$ children in the household) nodes with this path. Moreover, the red node of a household traversed by a transmission path is followed by another red node on the path (in another household) if it is the only node of that household on the transmission path, whereas it is followed by a blue node (in the same household) if two nodes of that household are on the transmission path. 
\end{remark}


\begin{lemma}
\label{lem:LS,LSp,trees}
In the RB tree network, the LS algorithm succeeds if and only if all nodes on the transmission path are symptomatic, and the LS+ algorithm succeeds if among the nodes of the transmission path, there exists a symptomatic node in each household, and the source is symptomatic.
\end{lemma}

\begin{remark}
We note that the statement for LS+ in Lemma \ref{lem:LS,LSp,trees} cannot be reversed, i.e., it is possible that LS+ succeeds even if among the nodes of the transmission path, there is a household with no symptomatic node (see Figure \ref{fig:LSp_cases}~(a)). Also, the proof of Lemma \ref{lem:LS,LSp,trees} does not hold if the LS+ algorithm proceeds to the next iteration at the first time $s_c'$ is updated (see Figure \ref{fig:LSp_cases}~(b)). Finally, in the proof of Lemma \ref{lem:LS,LSp,trees}, we do not make any assumptions about asymptomatic patients having had the disease previously or not, which implies that we could treat non-complying agents as asymptomatic patients without jeopardizing the correctness of the algorithms. 
\end{remark}

\subsubsection{The $\mathrm{DDE}_{\mathrm{NR}}$ Model}
\label{subsub:DDENR}
Focusing on tree networks is an important step towards making our models tractable for theoretical analysis, but it will not be enough; we will make two minor simplifications to the DDE model as well: we eliminate (i) the pre-symptomatic state and (ii) the recovered state, and we call the new model $\mathrm{DDE}_{\mathrm{NR}}$ (where NR stands for No Recovery). (i) The first assumption can be made without loss of generality, because the pre-symptomatic state does not have any effect on the disease propagation, nor on the success of the source detection algorithm. Indeed, according to Lemma \ref{lem:LS,LSp,trees}, the success of the LS and LS+ algorithms depends only on the information gained about the transmission path, and by the time of the first hospitalization, every node on the transmission path must have left the pre-symptomatic state (since we always have $T_P<T_E+T_H$), even if we include it in the model. (ii) The second assumption on the absence of recovery states amounts to take $T_I \rightarrow \infty$, which does have a small effect on the disease propagation, however, this effect is minimal
because $T_I=14$ is already quite large, and because only the very early phase of the infection is interesting for computing the success probabilities of the algorithms. Finally, this last assumption has no effect on the information gained by the algorithm since we assumed that recovered patients (who were symptomatic) can remember and reveal their symptom onset time in the same way as symptomatic infectious patients.

\subsubsection{Success Probability of LS}

Assuming that the distribution of length of the transmission path is provided for us (we give an approximation in Section \ref{sec:depth_dist}), the success probability of LS can be computed succinctly. We need a short definition before stating our result.

\begin{definition}
\label{def:p}
Let $p$ be the probability that a node is asymptomatic conditioned on the event that it is not hospitalized.
\end{definition}
A simple computation shows that
\begin{equation}
    \label{C5}
p=\P ( v \text{ is asy} \mid v \text{ is not hosp})=\frac{p_a}{p_a +(1-p_a)(1-p_h)}.
\end{equation}

\begin{lemma}
\label{lem:LSsucc}
For the $\mathrm{DDE}_{\mathrm{NR}}$ epidemic model with parameters $(p_i, p_a, p_h)$ on the RB tree with parameters $(d_c, d_h)$, and with $p$ computed in equation \eqref{C5}, we have
\begin{equation}
\label{eq:LS_success_RB_tree}
     \P(LS \textrm{ succeeds}) = \sum_{n=0}^{\infty} \left(1-p\right)^n \P(d(s,h) = n).
\end{equation}
\end{lemma}

\begin{proof}
Let us reveal the randomness that generates the epidemic in a slightly modified way than in the definition (Sections \ref{subsub:DDE} and \ref{subsub:DDENR}). As before, at the beginning only the source is infectious, and depending on course of the disease, the source can be symptomatic and hospitalized, symptomatic but not hospitalized, or asymptomatic with probabilities $(1-p_a)p_h, (1-p_a)(1-p_h), p_a$, respectively. In each moment, each infectious node infects each of its susceptible neighbors with probability $p_i$. If a node is infected, we reveal the information whether it will become hospitalized (which happens with the probability $(1-p_a)p_h$), but if it does not become hospitalized, we do not reveal whether the node is asymptomatic or symptomatic yet. Indeed, this information is not necessary for continuing the simulation of the epidemic since we assumed that there is no difference between the infection probabilities of symptomatic and asymptomatic nodes. Thereafter, when the first hospitalized case occurs, we reveal for each infected node $v$ on the transmission path (except the last node, which we know is hospitalised; see Definition \ref{def:inf_path}) whether it is asymptomatic or not. The only information we have about these nodes is that they are not hospitalized, which implies that the probability that a node is revealed to be asymptomatic on the transmission path is exactly the probability $p$ from Definition \ref{def:p} computed in \eqref{C5}.

By Lemma \ref{lem:LS,LSp,trees}, LS succeeds if and only if each node on the transmission path is symptomatic. Conditioning on the length of the transmission path, we can compute the probability of each node being symptomatic by equation \eqref{C5} as
\begin{equation}
    \P(LS\textrm{ suceeds} | d(s,h) = n)=\left(1- \P(v \textrm{ is asy}| v\textrm{ is not hosp}) \right)^n=\left(1-p\right)^n,
\end{equation}
from which \eqref{eq:LS_success_RB_tree} follows immediately.
\end{proof}

\subsubsection{Success Probability of LS+}

Computing the success probability of the LS+ algorithm is far more challenging compared to the LS algorithm, even if the distribution of the length of the transmission path is provided to us. Indeed, since the LS+ algorithm does further testing on the contacts and household members of asymptomatic nodes, it is essential to have additional information about the number of households on the transmission path. We give our main result on the LS+ in the next theorem, which we prove in Appendix \ref{sec:thrm:LSp_suc}.

\begin{theorem}
\label{thrm:LSp_suc}
Let $p$ be as in \eqref{C5} and let $\mathcal{S}(n,\alpha,\beta)$ be the set of $k$ integer values such that $k$ and $n$ have different parity and $n+1 - 2(\alpha + \beta) \geq k \geq 2-(\alpha+\beta)$. Then, for the $\mathrm{DDE}_{\mathrm{NR}}$ epidemic model with parameters $(p_i, p_a, p_h)$ on the RB tree with parameters $(d_c, d_h)$, we have
\begin{align}
    &\P(LS+ \textrm{ succeeds}) \ge \P(d(s,h) = 0) + (1-p)\P(d(s,h)=1)+ \nonumber \\
    & \sum_{n = 2}^{\infty} \sum_{\substack{\alpha,\beta \in \{0,1\} \\ k \in \mathcal{S}(n,\alpha,\beta)}}
    \binom{\frac{n+k-3}{2}}{k-2+\alpha+\beta}
    \frac{(d_h(1-p))^{\frac{n+k-1}{2}}(d_c(1+p))^{\frac{n-k+1}{2}-\alpha-\beta}d_c(d_c-1)^{k+\alpha + \beta-2}}{\lambda_1\left( \frac{d_c-1+D}{2}\right)^n + \lambda_2\left( \frac{d_c-1-D}{2}\right)^n} \P(d(s,h)=n),
\end{align}
where
\begin{align}
    D &= \sqrt{(d_c-1)^2 +4d_cd_h} \\
    \lambda_1 &= \frac{(d_c+1+D)(2d_h+d_c-1+D)}{2D(d_c-1+D)} \\
    \lambda_2 &= \frac{(D-d_c-1)(2d_h+d_c-1-D)}{2D(d_c-1-D)}.
\end{align}
\end{theorem}

\subsection{Approximating the Depth of the Path to the First Hospitalized Node}
\label{sec:depth_dist}

Section \ref{sec:lsp_succes} was dedicated to the success probability of the LS and LS+ algorithms, however, in these results, we are still missing the distribution of the transmission path length. In this subsection we address this problem by introducing simpler approximate models. 

\subsubsection{$(d_r, d)$-ary Random Exponential Tree}
\label{subsub:RET}
When we introduced the $\mathrm{DDE}_{\mathrm{NR}}$ model in Section \ref{subsub:DDENR}, we removed both parameters $T_P$ and $T_I$ from the DDE model (by removing the presymptomatic and the recovered states, respectively), but we kept the parameter $T_E$. In this step we will rescale the time parameter to make $T_E'=1$ by changing $p_i'$ to be $1-(1-p_i)^{T_E}$. Since we had $T_E=3$ by default, using $T_E'$ and $p_i'$ instead of $T_E$ and $p_i$ means that we choose 3 days to be our time unit, and the probability of infection is scaled to be the probability that the infection is passed in at least one of three days (since the RB tree is time-independent, if two nodes are connected, the infection can spread on it every day). We drop the prime from $p_i'$ and $T_E'$ for ease of notation. As a second approximation, instead of keeping track of two types of nodes (red and blue) as it is done in the RB tree, we propose to change our network model to an infinite $d$-regular tree, where $d$ is set to be the average degree of an RB tree. 

By making these two changes (tracking time at a coarser scale and simplifying the network topology to a $d$-regular tree), the growth of the epidemic becomes equivalent to a known model, the $d$-ary Random Exponential Tree ($d$-RET). Binary RETs have been introduced in \cite{feng2018profile}
. We give the definition below for completeness.
\begin{definition}
\label{def:dRET}
A $d$-ary Random Exponential Tree ($d$-RET) with parameters $d,p_i$ at time day $t$, denoted by $G_t(s)$, is a random tree rooted at node $s$. At day $0$, the tree $G_t(s)$ only has its root node~$s$. Let $\bar{G}_t(s)$ be the closure of $G_t(s)$, which is obtained by attaching external nodes to $G_t(s)$ until every internal node (a node that was already present in $G_t(s)$) has degree exactly $d$ in the graph $\bar{G}_t(s)$. Then, $G_{t+1}(s)$ is obtained from $\bar{G}_{t}(s)$ by retaining each external node with probability $p_i$, and dropping the remaining external nodes.
\end{definition}
Indeed, each node of a $d$-RET infects a new node with probability $p_i$ each day, and after a sufficiently long time, the $d$-RET becomes close to a large $d$-ary tree. Of course, we do not want to let the $d$-RET grow for a very long time, we only want it to grow until the first hospitalization occurs. So far we have not talked about the course of the disease of the nodes in the $d$-RET model because we could define the spread of the infection without it. Since we still need to do one final simplification to compute the distribution of the transmission path, we defer the discussion about hospitalizations, and how the parameters $p_a$ and $p_h$ are part of the model, to Section~\ref{subsub:DET}. Note that by considering the $d$-RET, we deviate from the idea of separating the epidemic and the network models; we only have a randomly growing tree, which is stopped at some time, when the tree is still almost surely finite. 

So far we only did simplifications to the model, which resulted in further and further deviations from the original version. Now we will make a small modification that brings our model back closer to the RB tree, without complicating the computations too much. We still make almost all maximum degrees of the RET uniform $d$, but we make an exception with the root, which will have maximum degree $d_r=d_c+d_h$. This makes the maximum degree of the root the same as the degree of the root of the RB tree. We call the resulting model a $(d_r,d)$-RET with parameter $p_i$. Since the close neighborhood of the source has a high impact on the success probability, we found that this solution gives the best results while keeping the computations tractable.

In our computations, only the profile the infection tree will be important, which motivates the next definition.
\begin{definition}
\label{def:total}
In the $(d_r,d)$-RET model with parameter $p_i$, let $A_{t,l}$ be the number of nodes during day~$t$ at level~$l$, and let $a_{t,l} = \mathbb{E}[A_{t,l}]$. 
Moreover, we define the random variable
\begin{align}
    &A_{t} = \sum_{t = 0}^{+\infty} A_{t,l} \label{eq:def_total}
\end{align}
with $A_{-1, l} = 0$ for all $l$, and its expectation $a_{t} = \mathbb{E}[A_{t}]$.
\end{definition}

As noted earlier, the $d$-RET model has only been analysed for $d=2$ to this date. We provide the expected number $a_{t,l}$ of nodes at level~$l$ in day~$t$ for the general case in the next theorem and corollary, which we prove in Appendices ~\ref{sec:theorem:a} and \ref{sec:corollary:at}.

\begin{theorem}
\label{theorem:a}
In the $(d_r, d)$-RET with parameter $p_i$, let $a_{t,l}$ be  as in Definition \ref{def:total}. Then
\begin{align}
    a_{t,0} &= 1 \\
    a_{t,l} &= d_rp_i\sum_{m = l-1}^{t-1} \binom{m}{l-1}(1-p_i)^{m-l+1}d^{l-1}p_i^{l-1} \textrm{, for } t \geq l \geq 1 \\
    a_{t,l} &= 0 \textrm{, for l > t}.
\end{align}
\end{theorem}

\begin{corollary}
\label{corollary:at}
In the RET$(p_i, d_r, d)$, let $a_{t}$ be the expectation of \eqref{eq:def_total}, as in Definition~\ref{def:total}. For $t\geq0$, 

\begin{equation}
    a_{t} = 1 + d_r\frac{(1-p_i+dp_i)^t - 1}{d-1}. \label{lemma:total}
\end{equation}
\end{corollary}

\subsubsection{Deterministic Exponential Tree with Parameters $p_a, p_h$ and $(c_{t,l})_{t,l \in \mathbb{N}}$}
\label{subsub:DET}

In the $(d_r,d)$-RET model it is still complicated to calculate the distribution of the depth of the first hospitalized node. For this reason, we approximate the RET model by a deterministic time-dependent tree with a prescribed profile. 
\begin{definition}
Let $(c_{t,l})_{t \in \mathbb{N} \bigcup \{-1\},l \in \mathbb{N}}$ be a two-dimensional array with $c_{t,l}=0$ for $t\in\{-1,0\}$ and $l\in \mathbb{N}$, except for $c_{0,0} = 1$, and with $c_{t,l} \geq c_{t,l-1}$ for any $t$ and any $l \geq 1$. Additionally, if we define $c_t = \sum_l c_{t,l}$, then the array $(c_{t,l})$ must satisfy $c_t> c_{t-1}$ for $t\ge0$. Then, we define the Deterministic Exponential Tree (DET) with parameter $(c_{t,l})_{t \in \mathbb{N} \bigcup \{-1\},l \in \mathbb{N}}$, as a time-dependent rooted tree, that has exactly $c_{t,l}$ nodes on level $l$ at time $t$. The edges between the adjacent levels are drawn arbitrarily so that the tree structure is preserved. 
\end{definition}

The formal assumptions on the array $(c_{t,l})$ are simply made to ensure that the DET starts with a single node at $t=0$, that it never shrinks on any level ($c_{t,l} \geq c_{t,l-1}$), and that it grows by at least one node in each time step ($c_t> c_{t-1}$).


We have defined the DET at any given time $t$, however, to determine the length of the transmission path, we are not interested in the DET at any given time, but only when the first hospitalization occurs. 
To compute the distribution of the first hospitalized node, we would like to have an absolute order on the times when the nodes are added, which we do by randomization. We say that on day $t$, nodes are added one by one to the DET, their order given by a uniformly random permutation, and each node is hospitalized with probability $(1-p_a)p_h$ (as in the original DDE model). When the first hospitalization occurs, we stop growing the tree, and we call the resulting (now random) model a stopped DET with parameters $(c_{t,l}), p_a, p_h$. We find the transmission path length distribution on the stopped DET in the next lemma, which we prove in Appendix~\ref{sec:lem:DET}.


\begin{lemma}
\label{lem:DET}
Let us consider the stopped DET model with parameters $(c_{t,l}), p_a, p_h$, and let $h$ denote the first hospitalized node. Then
\begin{align}
\label{eq:DET_lemma}
    \P(d(s,h) = l) = \sum_{t=0}^{+\infty} \frac{c_{t,l}-c_{t-1,l}}{c_t-c_{t-1}}  (1-(1-p_a)p_h)^{c_{t-1}}\left(1-(1-(1-p_a)p_h)^{c_t-c_{t-1}}\right).
\end{align}
\end{lemma}

We would like to set $c_{t,l}$ so that the DET is close to the RET described in Section~\ref{subsub:RET}. For equation~\eqref{eq:DET_lemma} to make sense, we should substitute integer values for $c_{t,l}$, however, for an approximation the equation can also be evaluated for fractional values as well.

\begin{remark}
\label{rem:subs}
If we substitute $c_{t,l}=a_{t,l}$ and $c_t=a_t$ in equation \eqref{eq:DET_lemma}, where $a_{t,l}$ is given in Theorem \ref{theorem:a} and $a_{t}$ is computed in Corollary~\ref{corollary:at}, then we get the expression
\begin{align}
\label{eq:rem_approx}
d_rp_i^{l-1}d^{l-1} \sum_{t=l}^{+\infty} \frac{\binom{t-1}{l-1}(1-p_i)^{t-l}}{(1-p_i+dp_i)^{t-1}}(1-(1-p_a)p_h)^{1 + d_r \frac{(1-p_i+dp_i)^{t-1} - 1}{d-1}} \left(1-(1-(1-p_a)p_h)^{d_r(1-p_i+dp_i)^{t-1})}\right),
\end{align}
which approximates the distribution of the transmission path length in the $(d_r, d)$-ary RET stopped at the first hospitalization.
\end{remark}

\section{Simulation Results}
\label{sec:simulation}

\subsection{Baseline Algorithms}
\subsubsection{Non-adaptive Baseline: Dynamic Message Passing}


There are few sensor-based source detection algorithms that are compatible with time-varying networks in the literature \cite{huang2017source,jiang2016rumor,fan2020identifying}. The most promising one among these algorithms \cite{jiang2016rumor} has a close resemblance to the a previous work of \cite{lokhov2014inferring} on Dynamic Message Passing (DMP) algorithms. Given the initial conditions on the identity of the source node and its time of infection, the DMP algorithm approximates the marginal distribution of the outcome of an epidemic at some later time $t$. The algorithm is exact on tree networks, and it computes a good approximation when there are not too many short cycles in the network. Therefore, the DMP algorithm can be used to approximate the likelihood of the observed symptom onset times for any (source,time) pair. Due to its flexibility, we were able to adapt the DMP algorithm to the SDCTF (see Appendix \ref{sec:DMP_all} for more details).

Originally, the DMP was applied to the source detection problem by computing the likelihood values for all possible (source,time) pairs, and then choosing the source node from the most likely pair as the estimate \cite{lokhov2014inferring}. However, testing all (source,time) pairs increases the time complexity of the algorithms potentially by a factor of $N^2$, which makes the algorithm intractable in many applications. 
Jiang et. al. \cite{jiang2016rumor} proposed a very similar algorithm to the DMP equations (which is unfortunately not exact even on trees), and solved the issue of intractability by a heuristic preprocessing step to the DMP algorithm. This preprocessing step, identifies a few candidate (source,time) pairs, by spreading the disease backward from the observations in a deterministic way (called reverse dissemination). Since we already approximate our data-driven model (DCS) by an epidemic model with deterministic transition times (DDE), it is natural for us to also implement the deterministic preprocessing step proposed by \cite{jiang2016rumor}. We produce 5 (source,time) pairs which are feasible for the 5 earliest symptom onset time observations (see Appendix~\ref{sec:DMP_feasible} for more details). It would have been ideal to run the algorithms for more than 5 pairs, but this was made impossible by the runtimes becoming very high. We run therefore our implementation of the DMP algorithm with the previously computed feasible (source,time) pairs as initial conditions to find the most likely source candidate.

The source estimation algorithms developed using the DMP algorithm do not specify how the sensors should be selected, and therefore place these non-adaptive sensors randomly. We refer to the resulting algorithm as random+DMP. The number of sensors is set so that it always exceeds the number that LS/LS+ would use. The simulation results are shown in Figure \ref{fig:baseline_plots} for the DDE+HNM model. Importantly, the deterministic preprocessing step of \cite{jiang2016rumor} is compatible with time-varying networks, which allows us to run the algorithm for the DCS+TU model as well (see Figure \ref{fig:DCSplots}).

\subsubsection{Adaptive Baseline: Size-Gain}

The Size-Gain (SG) algorithm was developed for epidemics which spread deterministically \cite{zejnilovic2015sequential}, and has been later extended to stochastic epidemics \cite{spinelli2017back}. It works by narrowing a candidate set based on a deterministic constraint. If $v_1, v_2$ are symptomatic observations, then $s_c$ is in the candidate set of SG if and only if
\begin{equation}
\label{eq:SG}
    |(t_{v_2}-t_{v_1})- (d(v_2,s_c) - d(v_1,s_c))| < \sigma(d(v_2,s_c)+d(v_1,s_c)),
\end{equation}
where $\sigma$ is the standard deviation of the infection time of a susceptible contact. If one of the observations, say $v_2$, is negative, then SG uses a condition almost identical to equation \eqref{eq:SG}, except that the absolute value is dropped, since a negative observation at time $t_{v_2}$ is only a lower bound on the true symptom onset time of $v_2$. These deterministic conditions are checked for every symptomatic-symptomatic or symptomatic-negative pair $(v_1, v_2)$ to determine if $s_c$ can be part of the candidate set. Next, SG places the next sensor adaptively at the node which reduces the candidate set by the largest amount in expectation (assuming a uniform prior on the source and its infection time), and it terminates when the candidate set shrinks to a single node. Note that the SG algorithm can fail if at least one of the deterministic conditions in equation~\eqref{eq:SG} is violated for some $(v_1, v_2)$ because of the randomness of the epidemic.

We use the existing implementation of the SG algorithm by \cite{spinelli2017back}, and adapt it to the SDCTF. We incorporate asymptomatic-symptomatic and asymptomatic-negative observations $(v_1,v_2)$ the same way as symptomatic-negative are incorporated; we drop the absolute value sign in equation \eqref{eq:SG}, because an asymptomatic observation at time $t_{v_1}$ is only an upper bound on the true symptom onset time of $v_1$. We impose the same daily limit to the number of sensors that can be placed by the SG algorithm in a single day as for the LS/LS+ algorithm, and if the candidate set size does not shrink to one on the day when both LS and LS+ have already provided their estimates, then the SG algorithm must make a uniformly random choice from the current candidate set as its source estimate. The simulation results are shown in Figure \ref{fig:baseline_plots} for the DDE+HNM model. We do not implement the SG algorithm for the DCS+TU model, because its runtime is too high, and because it is not clear how it should be implemented for time-varying networks.

\subsection{Comparison with Baselines}
\label{sec:comp_baseline}

\begin{figure}[h]
\begin{center}
 \includegraphics[width=\textwidth]{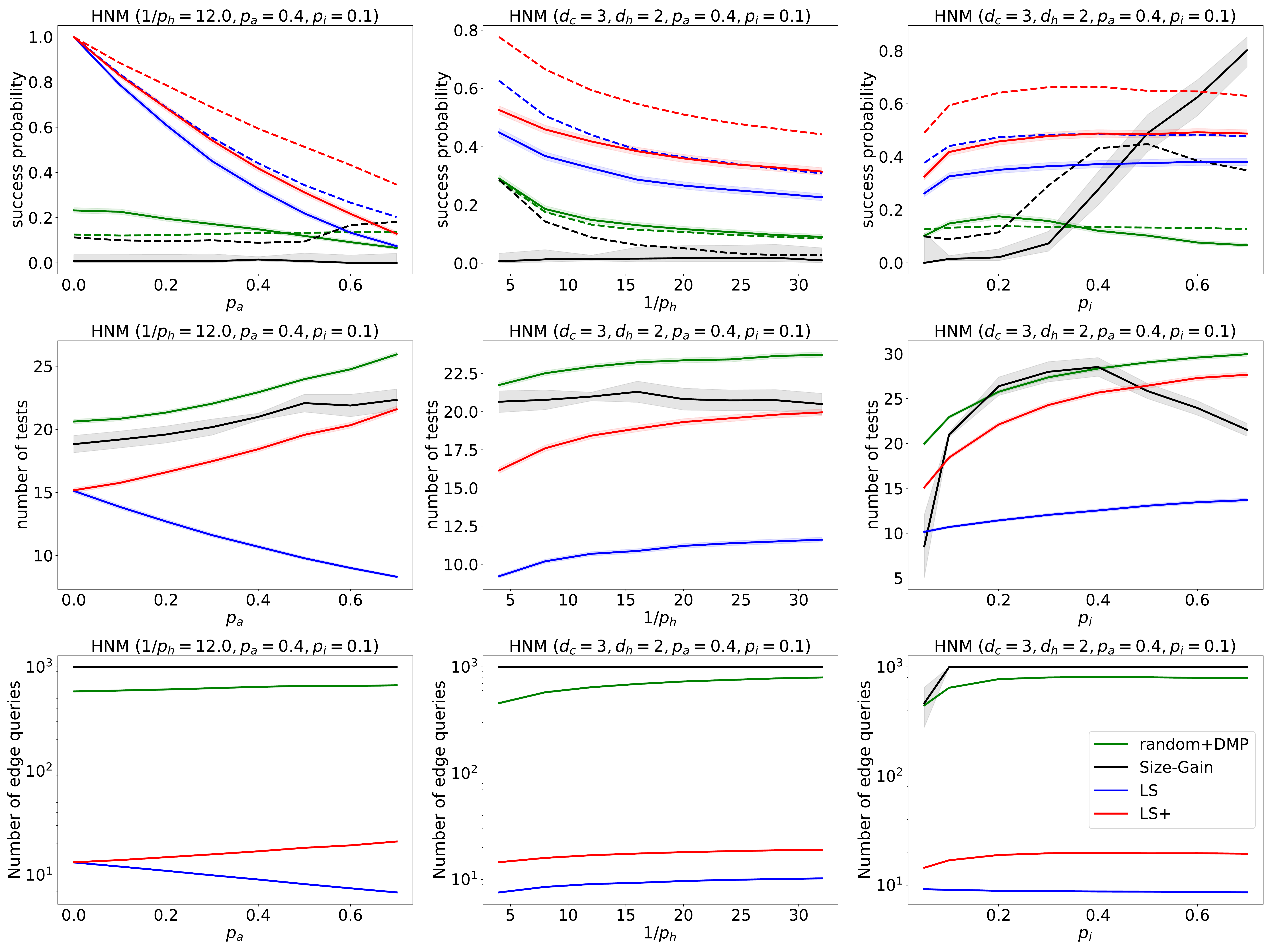}
  \caption{The performance of the algorithms LS, LS+, R and SG if the metric is the probability of finding the source (solid curves) or the first symptomatic patient (dashed curves). The simulations were computed on a population of $n=400$ individuals in the DDE model on the HNM, and each datapoint is the average of $4800$ independent realizations except for the SG algorithm, which was run with $192$ independent realizations. The confidence intervals for the success probabilities are computed using the Wilson score interval method, and for the tests and the queries using the Student's t-distribution.}
  \label{fig:baseline_plots}
  \end{center}
\end{figure}

We show our simulation results comparing the random+DMP, SG, LS and LS+ algorithms in Figure~\ref{fig:baseline_plots}. In the first row of Figure \ref{fig:baseline_plots}, we show the accuracy of the algorithms with solid curves. Since the LS/LS+ algorithms cannot detect an asymptomatic source, we also show what the accuracy would look like if the goal of the SDCTF was to detect the first symptomatic agent with dashed lines. It is clear that in both metrics and across a wide range of parameters, the LS+ algorithm performs best, followed by LS, next random+DMP, and finally SG. The only exception is for high values of $p_i$, where SG performs best. The good performance of SG for these parameters is expected, because SG was originally developed for deterministically spreading epidemics (i.e., $p_i=1$).

In the second row of Figure \ref{fig:baseline_plots}, we show the number of test/sensor queries used by the algorithms. LS uses the fewest tests, followed by LS+. The random+DMP and SG algorithms always use more tests than LS/LS+, except for large values of $p_i$. Finally, in the last row of Figure \ref{fig:baseline_plots} we show the number of contact (or in this case edge) queries used by the algorithms. Again, LS uses fewer queries than LS+, while both the random+DMP and SG algorithms query essentially the entire network.

Figure \ref{fig:baseline_plots} shows that the LS/LS+ algorithms are fairly robust to changes in the parameters of the model, except for the parameter $p_a$. Indeed, if there are many asymptomatic nodes in the network, then source detection becomes very challenging. It may be surprising that as $p_a$ grows, the number of tests that LS uses decreases, contrary to LS+. This is because as $p_a$ grows, the LS algorithm gets stuck more rapidly, while the LS+ algorithm compensates for the presence of asymptomatic nodes by using more test/sensor queries.

\subsection{Comparison of Simulations and Theoretical Results}
\label{sec:comp_theory}

\begin{figure}[h]
\begin{center}
 \includegraphics[width=\textwidth]{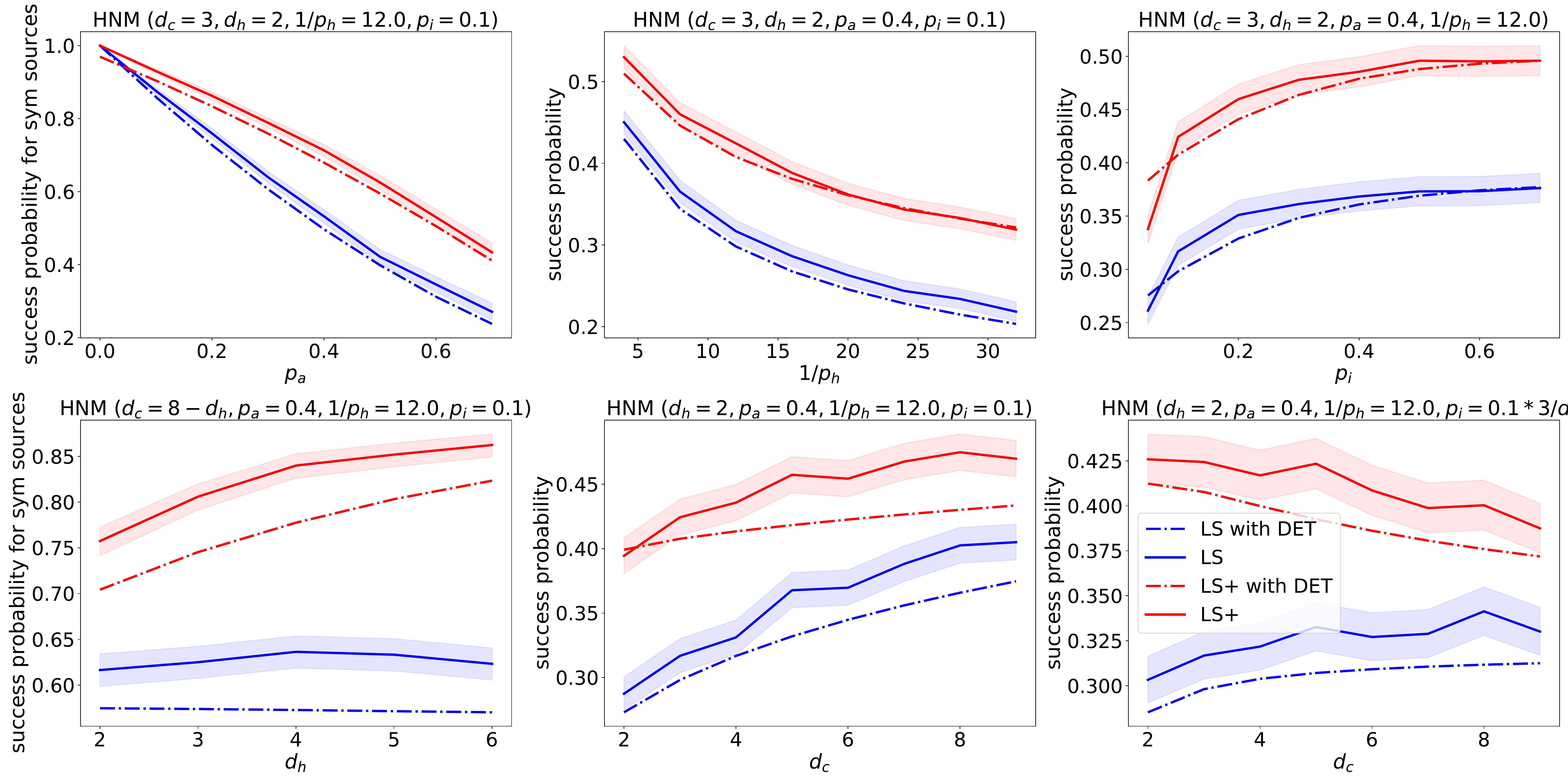}
  \caption{The probability of success of the LS and LS+ algorithms (solid curves) and their theoretical estimate (dash-dotted curves) with the success probabilities computed in Lemma \ref{lem:LSsucc} and Theorem \ref{thrm:LSp_suc}, while the transmission path distribution computed in equation \eqref{eq:rem_approx}. The simulation results were generated using the DDE model on HNM networks of size $n=1000$ with $4800$ independent samples. The $95\%$ confidence intervals are computed using the Wilson score interval method.}
  \label{fig:theory_plots}
  \end{center}
\end{figure}

The analytic results from Section~\ref{sec:theory} are in good agreement with the simulation results in Figure \ref{fig:theory_plots}. We also experiment with changing the parameters $d_h$, $d_c$ while keeping all the parameters fixed, and with changing $d_c$ while keeping the product $d_cp_i$ fixed. We observe that LS is not affected by the parameter $d_h$, whereas LS+ performs better with a higher $d_c$, which is expected because LS+ leverages the household structure of the network to improve over LS. Somewhat surprisingly, we also observe that a higher $d_c$ also improves the performance of both algorithms. This can be explained by the fact that a larger $d_c$ implies that there are more nodes in the close neighborhood of the source, which results in shorter transmission paths, making source detection less challenging. Finally, if we increase $d_c$ but keep $d_cp_i$ fixed, the performance of the algorithms does not change as much, which confirms the intuition that it is the number of infections caused by an infectious node in a single day that matters the most (as we discussed in Section~\ref{sec:boe_sec}).

\subsection{Simulations on the DCS Model}
\label{sec:simDCS}

\begin{figure}[h]
\begin{center}
 \includegraphics[width=\textwidth]{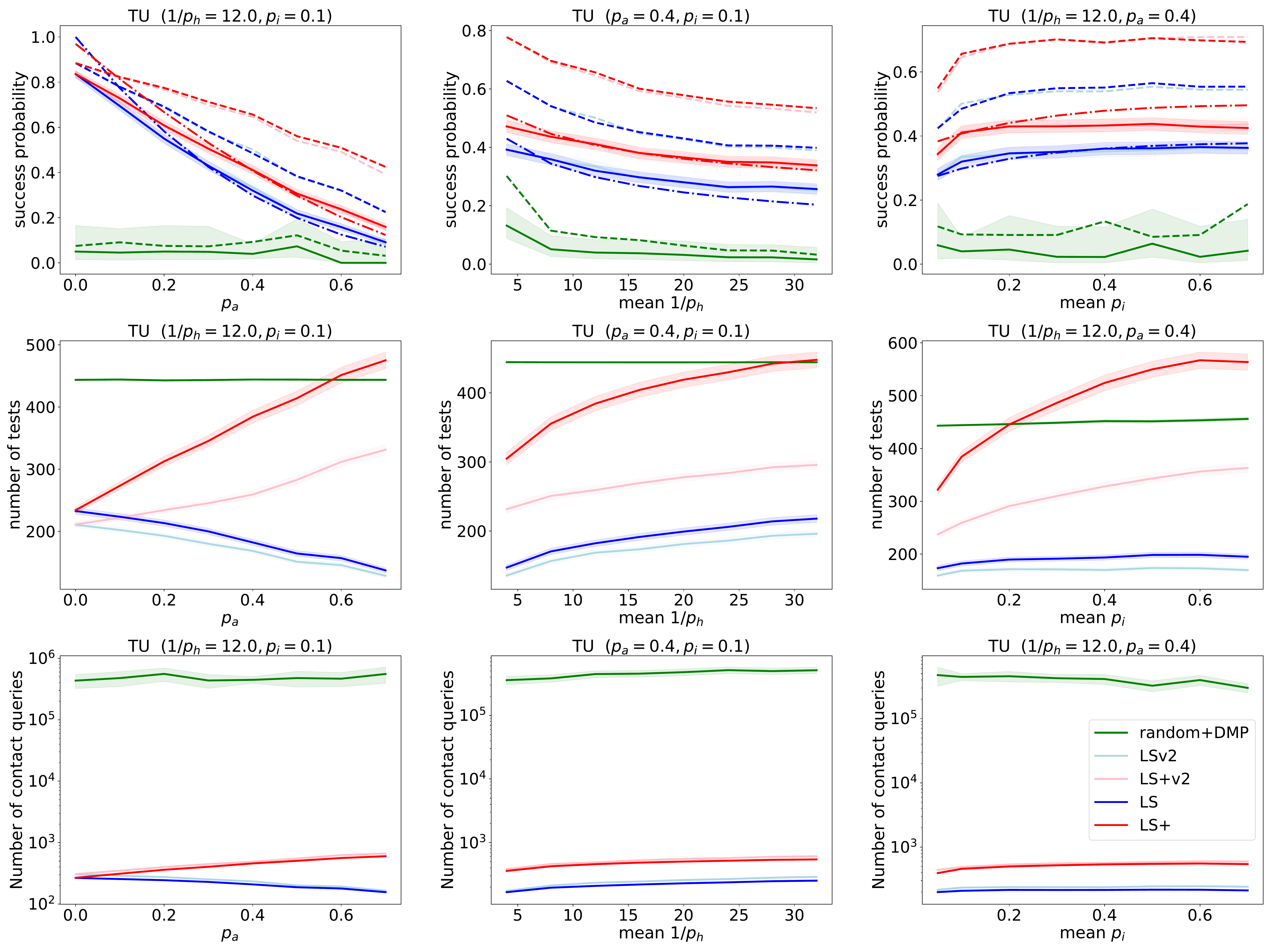}
  \caption{The performance of the algorithms LS, LS+ and random+DMP on the DCS model with the Tubingen dynamics if the metric is the probability of finding the source (solid curves) or the first symptomatic patient (dashed curves), together with the theoretical results (dash-dotted lines), as shown in Figure \ref{fig:theory_plots}. The simulations were computed on a population of $n=9054$ individuals, and each datapoint is the average of $2400$ independent realizations for the LS/LS+/LSv2/LS+v2 algorithms, and $48$ independent realizations for the random+DMP algorithm. The default population and infection parameters were selected to match the population and COVID-19 infection datasets of Tubingen. The confidence intervals for the success probabilities are computed using the Wilson score interval method, and for the tests and the queries using the Student's t-distribution.}
  \label{fig:DCSplots}
  \end{center}
\end{figure}

We show our simulation results on our most realistic DCS+TU model in Figure \ref{fig:DCSplots}. We make very similar observations on this model as the ones that we have made on the DDE+HNM model in Sections~\ref{sec:comp_baseline} and \ref{sec:comp_theory}, which shows that the LS/LS+ algorithms and our analysis of their performance is robust to changes in the epidemic and network models.

In the DCS+TU model, we used a fixed limit on the number of sensors that the random+DMP model selects, instead of setting the limit based on the LS+ algorithm. As a result, for a few parameters the LS+ algorithm used more tests than the random+DMP model. However, we note that by updating the candidate node immediately after an earlier symptom onset time is revealed (see Section \ref{sec:localsearch}), we can essentially cut the number of required tests for the LS+ algorithm by half (LSv2 and LS+v2), without sacrificing the performance of the algorithms.

\section{Discussion}

\label{sec:discussion}

We have introduced the LS and LS+ algorithms in the SDCTF, and we have used a sequence of models on which we can compute their accuracy (probability of finding the correct source) rigorously. We find that both LS and LS+ outperform baseline algorithms, even though the baselines essentially query all contacts on a transmission path between agents, while LS and LS+ query only a small neighborhood of the source. One could argue that LS and LS+ beat the baseline algorithms only because we benchmark them in our own framework, which is different from the framework for which the baseline algorithms were developed. However, we argue that the LS/LS+ algorithms are robust to changes in the framework due to their simplicity, and we support our argument by simulation results. The runtimes of the LS/LS+ algorithms are also much lower than the baselines and do not depend on the network size since they are local algorithms - as opposed to the baselines, which have quadratic or even larger dependence on the network size. The ``low-tech'' approach in the design of the LS/LS+ algorithms increases their potential to be implemented in real-world scenarios, possibly even in a decentralized way, similarly to contact tracing smart phone applications \cite{troncoso2020decentralized}, which is an interesting direction for future work.

\bibliographystyle{ACM-Reference-Format}  
\bibliography{references}  


\appendix

\section{Additional Proofs}

\begin{figure}[h]
\begin{center}
 \includegraphics[width=\textwidth]{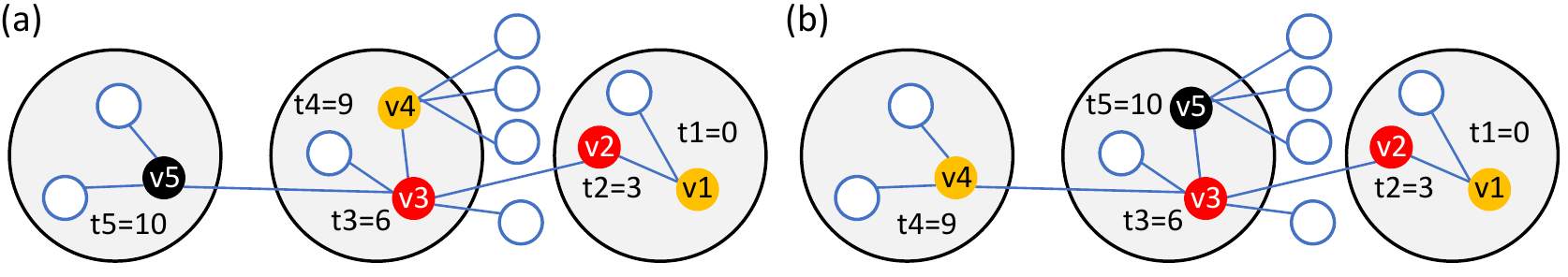}
  \caption{Illustration for Lemma \ref{lem:LS,LSp,trees} using the same coloring as Figure \ref{fig:mobility_models}~(a). (a) An example for an epidemic where among the nodes of the transmission path $(v_1, v_2, v_3, v_5)$, the middle household contains no symptomatic node (only the asymptomatic node $v_3$), but the LS+ algorithm still succeeds. Indeed, at iteration 0 we set $s_{c,0}=v_5$, after which we find that $v_3$ is asymptomatic, and next that $v_2$ is asymptomatic and $v_4$ is symptomatic, with a lower symptom onset time then $v_5$. Hence, in iteration 1 we set $s_{c,1}=v_4$, and we find that $v_3, v_2$ are asymptomatic and $v_1$ is is symptomatic, with a lower symptom onset time then $v_4$. Finally, in iteration 2 we set $s_{c,2}=v_1$, and we find $s_c'=v_1=s_{c,2}$, which implies that the algorithm stops, and returns the correct source $v_1$. (b) An example for an epidemic where the LS+ algorithm would fail if we would update the candidate before the test queue becomes empty. Similarly to subfigure (a), in iteration 0 of the algorithm first learns about asymptomatic node $v_3$ and next about asymptomatic node $v_2$ and symptomatic node $v_4$. If the algorithm updates the candidate to $v_4$ and continues further, instead of scheduling the tests of the household members of $v_2$, then it is not hard to check that $v_4$ will be the final estimate and the algorithm fails. However, if the algorithm waits until the test queue becomes empty and tests the household members of $v_2$, then $v_1$ becomes the next candidate and the algorithm succeeds.}
  \label{fig:LSp_cases}
  \end{center}
\end{figure}

\subsection{Proof of Lemma \ref{lem:LS,LSp,trees}}
\label{sec:LS,LSp,trees_app}

We start by restating the lemma here for convenience.

\begin{lemma}
\label{lem:LS,LSp,trees_app}
In the RB tree network, the LS algorithm succeeds if and only if all nodes on the transmission path are symptomatic, and the LS+ algorithm succeeds if among the nodes of the transmission path, there exists a symptomatic node in each household, and the source is symptomatic.
\end{lemma}

\begin{proof}
Throughout the proof we assume that there is no limitation on the number available tests. We can make this assumption because in the SDCTF there is only a daily limit on the number tests, there is no limitation on the number of days, and neither the LS nor the LS+ algorithms proceed in an iteration until the test queue becomes empty, which implies that all nodes that enter the test queue get eventually tested.

Suppose that the LS algorithm succeeds. Then the list of candidate nodes $s_c$ at different iterations forms a path that consists entirely of symptomatic nodes between the source and the first hospitalized node. In tree networks, the transmission path is the only path between the source and the first hospitalized node, which yields the "only if" part of the statement on the LS algorithm.

Next, suppose that all nodes on the transmission path are symptomatic. Then, we claim that the candidate node $s_{c,i}$ computed in the $i^{th}$ iteration of the LS algorithm is $v_{l-i}$, the $i^{th}$ node of the reverse transmission path. Our claim is definitely true for $i=0$, because $s_{c,0}$ is initialized to be the first hospitalized node $v_l$. Then, the proof proceeds by induction. By the induction hypothesis, in the $i^{th}$ step, $s_{c,i}=v_{l-i}$, and since we are on a tree, the symptom onset time of $v_{l-(i+1)}$ (which is revealed because all nodes on the transmission path are symptomatic by assumption) is the only symptom onset time among the neighbors of $s_{c,i}$ that have a lower symptom onset time than $s_{c,i}$ itself. Therefore $s_c'=v_{l-(i+1)}$, and $s_{c,i+1}$ is updated to be $v_{l-(i+1)}$ in the beginning of the next iteration, which proves that the induction hypothesis holds until the source is reached.

Finally, suppose that among the nodes of the transmission path, there exists a symptomatic node in each household, and the source is symptomatic. Let us denote by $w_i$ the $i^{th}$ symptomatic node of the \emph{reverse} transmission path. Then, we claim that the candidate list $s_{c,i}$ computed in the $i^{th}$ iteration of the LS+ algorithm equals $w_i$. Similarly to the case of the LS algorithm, the $i=0$ case holds by definition, and we proceed by induction. Suppose that $s_{c,i}=w_i$. It will also be useful to define the index of $w_i$ on the \emph{forward} transmission path (without skipping asymptomatic nodes). Let $j$ be this index, for which therefore $w_i=v_j$. Now we distinguish 3 cases: (i) $v_{j-1}=w_{i+1}$ is symptomatic, (ii) $v_{j-1}$ is asymptomatic and $v_{j-2}=w_{i+1}$ is symptomatic, and (iii) 
$v_{j-1}$ and $v_{j-2}$ are asymptomatic and $v_{j-3}=w_{i+1}$ is symptomatic. We claim that there are no more cases, and that in all three cases $w_{i+1}$ is tested in the $i^{th}$ iteration of the LS+ algorithm. Case (i) is immediate because all neighbors of $s_{c,i}$ are tested. Case (ii) is only possible if either $v_{j-1} \in H(s_{c,i})$ or $v_{j-2} \in H(v_{j-1})$, otherwise $v_{j-1}$ would be a lone asymptomatic node in a household, which contradicts the assumption that there is a symptomatic node in each household. Since all the contacts of asymptomatic nodes in $H(s_{c,i})$ (see Figure \ref{fig:Alg1_figure}~(d)) and all nodes in the household of asymptomatic nodes are tested in the LS+ algorithm (see Figure \ref{fig:Alg1_figure}~(e)), $v_{j-2}$ must be tested too. Finally, case (iii) is possible only if $v_{j-1} \in H(s_{c,i})$ and $v_{j-3} \in H(v_{j-2})$ both hold, otherwise $v_{j-1}$ or $v_{j-2}$ would be a lone asymptomatic node in a household. Similarly to the previous case, $v_{j-3}$ must be tested (see Figure \ref{fig:Alg1_figure}~(f)). There are no more cases because, by Remark \ref{rem:RBtree}, on the RB tree a transmission path can only have two nodes in each household, and we assumed that there exists a symptomatic node in each household among the nodes of the transmission path.

After we proved that $w_{i+1}$ is tested in the $i^{th}$ iteration of the LS+ algorithm, we must still show that it will be the next candidate $s_{c,i+1}$ for the induction hypothesis to hold. This is true because once the symptom onset time of $w_{i+1}$ is revealed, none of its neighbors are scheduled for testing, and therefore all tested nodes have $w_{i+1}$ on their path to the source, which means that $w_{i+1}$ must have the lowest revealed symptom onset time, and therefore that it will be the next candidate~$s_{c,i+1}$.
\end{proof}

\subsection{Proof of Theorem \ref{thrm:LSp_suc}}
\label{sec:thrm:LSp_suc}

We are going to need prove a few intermediate results before proving Theorem \ref{thrm:LSp_suc}. A first step is to count all the possible paths from the source with a given length. 

\begin{definition}
Let $G(s)$ be the RB tree with parameters $(d_c, d_h)$, and let $s$ be the source. A Red-Blue (RB) path of length $n$ is any path of nodes in $(s=v_0, v_1, ...v_n)$ such that $(v_i, v_{i+1}) \in E'$ for $0\leq i < n$. Let $\mathcal{C}_n$ be the set of RB paths of length $n$.
\end{definition}

\begin{lemma}
\label{lem:Cn}
In the RB tree with parameters $(d_c, d_h)$, $|C_0| = 1$, while for $n \geq 1$,
\begin{equation}
    |\mathcal{C}_n| = \lambda_1\left( \frac{d_c-1+D}{2}\right)^n + \lambda_2\left( \frac{d_c-1-D}{2}\right)^n
\end{equation}
where
\begin{align}
    D &= \sqrt{(d_c-1)^2 +4d_cd_h} \\
    \lambda_1 &= \frac{(d_c+1+D)(2d_h+d_c-1+D)}{2D(d_c-1+D)} \\
    \lambda_2 &= \frac{(D-d_c-1)(2d_h+d_c-1-D)}{2D(d_c-1-D)}.
\end{align}
\end{lemma}

\begin{proof}
Let us keep track of the number of RB paths of length $n$ depending on the color of the last node in the path. Let $r_n$ and $b_n$ be the numbers of RB paths of length $n$ such that the last node is red and blue, respectively. A RB path of length 0 consists only of the source, which implies that $r_0 = 1$ and $b_0=0$. The source has $d_c$ red and $d_h$ blue neighbours, which implies that $r_1 = d_c$ and $b_1 = d_h$. 

Suppose that $P$ is an RB path of length $n \geq 2$. If the last node of $P$ is red, then the node before the last node can be both blue or red. Red nodes other than the source have $d_c-1$ red children, while blue nodes have $d_c$ red children, yielding
\begin{align}
\label{eq:rn}
    r_n = (d_c-1)r_{n-1}+d_cb_{n-1}, \textrm{ for }n\geq 2.
\end{align}
If the last node of $P$ is blue, then the node before has to be red. Since every red node, including the source, has $d_h$ blue children, we have
\begin{align}
    \label{eq:bn}
    b_n = d_hr_{n-1}, \textrm{ for }n\geq 1 .
\end{align}
By substituting equation \eqref{eq:bn} into equation \eqref{eq:rn}, we obtain the recurrence
\begin{align}
    r_n = (d_c-1)r_{n-1}+d_cd_hr_{n-2}, \textrm{ for  }n\geq 2.
\end{align}
We solve this recurrence equation by calculating the characteristic equation
\begin{equation}
    t^2 - (d_c-1)t - d_cd_h = 0,
\end{equation}
whose roots are 
\begin{align}
    t_1 &= \frac{d_c-1+\sqrt{(d_c-1)^2 +4d_cd_h}}{2} = \frac{d_c-1+D}{2} \\
    t_2 &= \frac{d_c-1-\sqrt{(d_c-1)^2 +4d_cd_h}}{2} =  \frac{d_c-1-D}{2} 
\end{align}
yielding the the general solution
\begin{equation}
    r_n = c_1t_1^n + c_2t_2^n,
\end{equation}
where  $c_1, c_2$ are given by the initial conditions for $n=0,1$
\begin{align}
    &c_1 + c_2 = r_0 = 1 \\
    &c_1 t_1 + c_2t_2 = r_1 = d_c,
\end{align}
which are
\begin{align}
    c_1 &= \frac{1}{2} + \frac{d_c+1}{2\sqrt{(d_c-1)^2+4d_cd_h}} = \frac{1}{2} + \frac{d_c+1}{2D} \\
    c_2 &= \frac{1}{2} - \frac{d_c+1}{2\sqrt{(d_c-1)^2+4d_cd_h}} =\frac{1}{2} - \frac{d_c+1}{2D}.
\end{align}
From equations \eqref{eq:rn} and \eqref{eq:bn} we conclude that for $n\geq 1$,
\begin{equation}
b_n = d_h(c_1t_1^{n-1} + c_2t_2^{n-1})
\end{equation}
and therefore
\begin{align}
    |C_n| = r_n + b_n = \lambda_1t_1^n + \lambda_2t_2^n,
\end{align}
where 
\begin{align}
    \lambda_1 &= c_1\left(1+\frac{d_h}{t_1}\right) \\
    \lambda_2 &= c_2\left(1+\frac{d_h}{t_2}\right) .
\end{align}
Inserting the values for $t_1, t_2, c_1, c_2$ we obtain the desired result.
\end{proof}

Since LS+ improves on LS by making use of the household structure of the network, we need further information about the household structure of the transmission paths. Recall that by Remark~\ref{rem:RBtree}, households on transmission paths on an RB tree were characterized either by a single red node (that is followed by a red node), or a pair of consecutive red and blue nodes. The following definition and lemma refine our previous result on counting the number of RB paths by taking the household structure into account.

\begin{definition}
Let $P=\{ s=v_0,v_1,\ldots, v_n=h \}$ be a RB path of length $n$. We say that a node $v$ on the path $P$ is in a $P$-single-household if no other node from $P$ is in the same household as $v$. Otherwise, we say $v$ is in a $P$-multi-household. Given a path $P$, let $M_s: \mathcal{C}_n \rightarrow \{0,1\}$ be the indicator function that the source is in a $P$-multi-household. Similarly, let $M_l: \mathcal{C}_n \rightarrow \{0,1\}$ be the indicator function that the last node of path $P$ is in a $P$-multi-household. Finally, for $0 \leq k \leq n+1$ and $\alpha, \beta \in \{0,1\}$, let
\begin{equation}
    C_{n,k,\alpha, \beta} = \{P \in \mathcal{C}_n : (\textrm{there are exactly }k \textrm{ nodes in } P-\textrm{single-households})\wedge (M_s(P)=\alpha) \wedge (M_l(P) = \beta)\}.
\end{equation}
\end{definition}

The set $C_{n,k,\alpha, \beta}$ depends on 4 parameters, but only some combinations of these parameters make it non-empty. The following definition will be useful in this regard.

\begin{condition}
\label{cond:path}
Let $\alpha, \beta \in \{0,1\}$ and $n \geq 2$. We say $k \in \mathbb{N}$ satisfies Condition~\ref{cond:path} if and only if $k$ and $n$ have different parity and $n+1 - 2(\alpha + \beta) \geq k \geq 2-(\alpha+\beta)$.
\end{condition}

\begin{lemma}
\label{lem:Cnkab}
It holds that $|C_{0,1,0,0}| = 1$, $|C_{1,0,1,1}| = d_h$ and $|C_{1,2,0,0}| = d_c$. Let $\alpha, \beta \in \{0,1\}$, let $n \geq 2$ and let $k\in \mathbb{N}$ satisfy Condition \ref{cond:path}. Then
\begin{equation}
\label{eq:Cnkab}
    |\mathcal{C}_{n,k,\alpha, \beta}| = \binom{\frac{n+k-3}{2}}{k-2+\alpha+\beta}d_h^{\frac{n-k+1}{2}}d_c^{\frac{n-k+3}{2}-\beta-\alpha}(d_c-1)^{k+\alpha + \beta-2}. 
\end{equation}
In all other cases $|C_{n,k,\alpha, \beta}| = 0$.
\end{lemma}

\begin{proof}
Since there are $n+1$ nodes on path $P$, with $k$ in $P$-single households and thus $n+1-k$ of them in $P$-multi-households, we must have 
$$k + \frac{n+1-k}{2}=\frac{n+k+1}{2}$$
households along path~$P$ in total. Clearly, the numbers $n$ and $k$ cannot be of the same parity for any RB path $P$, which is thus assumed for the rest of the proof (this assumption is also part of Condition~\ref{cond:path}).

If $n=0$, then the source is also the first hospitalized node, and it is in a $P$-single-household, which implies that $|C_{0,1,0,0}| = 1$. If $n=1$, then there are two cases: either the source is in the same $P$-multi-household with the first hospitalized node, or both of them are in $P$-single-households. The former case is possible via $d_h$ edges from the source, which gives $|C_{1,0,1,1}| = d_h$, while the latter case is possible via $d_c$ edges, and gives $|C_{1,0,1,1}| = d_c$. Since these are the only possible RB paths of length $n\le 1$, we must have $|C_{0,k,\alpha,\beta}|=|C_{1,k,\alpha,\beta}| = 0$ for any other choice of parameters $k,\alpha$ and $\beta$.

Let us assume that $n\geq 2$. Then, the source and the first hospitalized node are not in the same household. Let us denote the household of the source by $H_s$ and the household of the first hospitalized node by $H_h$. Note that $(1-\alpha)$ and $(1-\beta)$ are the indicators of $H_s$ and $H_h$ being $P$-single-households, and therefore $k \geq (1-\alpha) + (1-\beta)$. If this inequality (which is also part of Condition \ref{cond:path}) does not hold, then clearly $|\mathcal{C}_{n,k,\alpha, \beta}|=0$. Similarly, the number of $P$-multi-households is $\frac{n-k+1}{2}$ and we must have $\frac{n-k+1}{2} \geq \alpha + \beta$ for $|\mathcal{C}_{n,k,\alpha, \beta}|>0$, which implies the inequality $n+1 - 2\alpha - 2\beta \geq k$. Therefore $\mathcal{C}_{n,k,\alpha, \beta}$ is empty if Condition \ref{cond:path} does not hold. For the rest of the proof we assume that Condition \ref{cond:path} does hold.

There are $\frac{n+k-3}{2}$ households along path~$P$, excluding $H_s$ and $H_h$. Among them, there are $k-(1-\alpha)-(1-\beta)$ $P$-single-households, which can be chosen in $\binom{\frac{n+k-3}{2}}{k-2+\alpha+\beta}$ ways. Once we know the color of each node along the path, the number of RB paths can be computed by multiplying the numbers of children with the appropriate color of each node.  $P$-single-households have no blue nodes, and $P$-multi-households have exactly one, which implies that there are $\frac{n-k+1}{2}$ blue nodes. Since blue nodes are preceded by red nodes that have $d_h$ blue children, they give the multiplicative factor $d_h^{\frac{n-k+1}{2}}$. Blue nodes, except from the first hospitalized node (if it is blue), have $d_c$ red children. So far we have accounted for all of the nodes in $P$-multi-households and none of the nodes in $P$-single-households. If the source is in a $P$-single-household, then we must count its red children, whose number is $d_c$. This implies that there exist $\frac{n-k+1}{2}-\beta+(1-\alpha)$ nodes with $d_c$ red children. Finally, each $P$-single-household, except $H_s$ and/or $H_h$ in case they are $P$-single households, has $d_c-1$ red children. There are $k-(1-\alpha)-(1-\beta)$ such $P$-single-households, which gives the final term in equation \eqref{eq:Cnkab}.
\end{proof}

The sets $\mathcal{C}_{n,k,\alpha, \beta}$ define equivalence classes on the transmission paths based on their household structure. In the next lemma we show that once we know which equivalence class we are in, it is possible compute the success probability of the LS+ algorithm.

\begin{lemma}
\label{lem:LSp|Cnkab}
Let $P$ be the transmission path in the $\mathrm{DDE}_{\mathrm{NR}}$ epidemic model with parameters $(p_i, p_a, p_h)$ on the RB tree with parameters $(d_c, d_h)$, and let $p$ be as computed in \eqref{C5}. Then, it holds that
$$\P(LS+ \textrm{ succeeds} | P \in \mathcal{C}_{0,1,0,0}) = 1$$ and $$\P(LS+ \textrm{ succeeds} | P \in \mathcal{C}_{1,0,1,1}) = \P(LS+ \textrm{ succeeds} | \mathcal{C}_{1,2,0,0}) = 1-p.$$ 
Let $\alpha, \beta \in \{0,1\}$, let $n \geq 2$ and let $k\in \mathbb{N}$ satisfy Condition \ref{cond:path}. Then, it holds that
\begin{equation}
    \P(LS+ \textrm{ succeeds} | P \in \mathcal{C}_{n,k,\alpha, \beta}) \ge (1-p)^{\frac{n+k-1}{2}}(1+p)^{\frac{n-k+1}{2}-\alpha-\beta}.
\end{equation}
In all other cases $\P(LS+ \textrm{ succeeds} | P \in \mathcal{C}_{n,k,\alpha, \beta})$ is not defined.
\end{lemma}

\begin{proof}
If $n = 0$, then $k = 1$ and $\alpha = \beta = 0$. In that case, the source is the first hospitalized node and LS+ always succeeds. If $n =1$, then the first hospitalized node is in the neighbourhood of the source, and LS+ succeeds if and only if the source is symptomatic, which happens with probability $1-p$.

Let us assume that $n \geq 2$ and that $k$ satisfies Condition \ref{cond:path} (otherwise $|\mathcal{C}_{n,k,\alpha, \beta}|=0$ and $\P(LS+ \textrm{ succeeds}| P \in \mathcal{C}_{n,k,\alpha, \beta})$ is not defined). By Lemma \ref{lem:LS,LSp,trees} the LS+ algorithm succeeds in the $\mathrm{DDE}_{\mathrm{NR}}$ model on the RB tree if, among the nodes of the transmission path, there exists a symptomatic node in each household, and the source is symptomatic, which means that we can prove a lower bound on the success probability of LS+. Let us assume that the source is indeed symptomatic. Since the first hospitalized node is symptomatic by definition, the households of the source and of the first hospitalized node cannot make the LS+ algorithm fail. Let us denote these two households by $H_s$ and $H_h$, respectively. Also, let $M$ and $S$ be the sets of all $P$-multi- and $P$-single-households, respectively, excluding $H_s$ and $H_h$. Then, LS+ succeeds if all nodes in the households of $S$ are symptomatic, and if at least one node in the households of $M$ is symptomatic, which has probability $1-p$ and $1-p^2$ for each type of household, respectively, by equation \eqref{C5}. These observations yield that

\begin{align}
    \P(LS+ \textrm{ succeeds} | P \in \mathcal{C}_{n,k, \alpha, \beta}) &\ge \P(\textrm{source is sym}) (1-p)^{|S|} (1-p^2)^{|M|} \nonumber \\
    &= (1-p)(1-p)^{k-2+\alpha+\beta}(1-p^2)^{\frac{n-k+1}{2}-\alpha-\beta} \nonumber \\
    &= (1-p)^{k-1+\alpha+\beta}(1-p^2)^{\frac{n-k+1}{2}-\alpha-\beta}.
\end{align}

\end{proof}

Finally, we are ready to state and prove Theorem \ref{thrm:LSp_suc} on the success probability of LS+, which we restate here for convenience.

\begin{theorem}
\label{thrm:LSp_suc_app}
Let $p$ be as in \eqref{C5} and let $\mathcal{S}(n,\alpha,\beta)$ be the set of $k$ values that satisfy Condition \ref{cond:path}. Then, for the $\mathrm{DDE}_{\mathrm{NR}}$ epidemic model with parameters $(p_i, p_a, p_h)$ on the RB tree with parameters $(d_c, d_h)$ we have
\begin{align}
    &\P(LS+ \textrm{ succeeds}) \ge \P(d(s,h) = 0) + (1-p)\P(d(s,h)=1)+ \nonumber \\
    & \sum_{n = 2}^{\infty} \sum_{ \substack{\alpha,\beta \in \{0,1\} \\ k \in \mathcal{S}(n,\alpha,\beta)}}
    \binom{\frac{n+k-3}{2}}{k-2+\alpha+\beta}
    \frac{(d_h(1-p))^{\frac{n+k-1}{2}}(d_c(1+p))^{\frac{n-k+1}{2}-\alpha-\beta}d_c(d_c-1)^{k+\alpha + \beta-2}}{\lambda_1\left( \frac{d_c-1+D}{2}\right)^n + \lambda_2\left( \frac{d_c-1-D}{2}\right)^n} \P(d(s,h)=n),
\end{align}
where $D,\lambda_1$ and $\lambda_2$ are terms depending on parameters $d_c$ and $d_h$ and are computed explicitly in Lemma \ref{lem:Cn}.
\end{theorem}

\begin{proof}
Let us extend the domain of $\P(LS+ \textrm{ succeeds} | P \in C_{n,k, \alpha, \beta})$ by function $g$ defined as $g: \mathbb{N} \times \mathbb{N} \times \{0,1\} \times \{0,1\} \rightarrow [0,1] $ such that
\begin{equation}
    g(n,k,\alpha, \beta) = \begin{cases}
 \P(LS+ \textrm{ succeeds} | P \in C_{n,k, \alpha, \beta})& \text{ if } k \in \mathcal{S}(n,\alpha,\beta) \\ 
 0& \text{ if } k \not\in \mathcal{S}(n,\alpha,\beta).
\end{cases}
\end{equation}
Unlike $\P(LS+ \textrm{ succeeds} | P \in C_{n,k, \alpha, \beta})$, $g$ is defined for every  4-tuple of parameters $(n,k, \alpha, \beta) \in \mathbb{N} \times \mathbb{N} \times \{0,1\} \times \{0,1\}$. By the law of total probability we expand the success probability by conditioning on the path $P$ being of length~$n$ as
\begin{align}
    \label{eq:LS+expand}
    \P(LS+ \textrm{ succeeds}) =& \sum_{n = 0}^{\infty} \sum_{k = 0}^{\infty} \sum_{\alpha,\beta \in \{0,1\}} g(n,k, \alpha, \beta) \P(P \in C_{n,k, \alpha, \beta})  \nonumber \\
    =& \sum_{n = 0}^{\infty}  \sum_{k = 0}^{\infty}  \sum_{\alpha,\beta \in \{0,1\}} g(n,k,\alpha, \beta) \P(P \in C_{n,k, \alpha, \beta} | P \in C_{n}) \P(d(s,h)=n).
\end{align}
Next, we exchange the sums over $\alpha, \beta$ and $k$. This allows us to sum over only those $k$ values that satisfy Condition \ref{cond:path}, which implies that $\P(LS+ \textrm{ succeeds} | P \in C_{n,k, \alpha, \beta})$ is well-defined. As in Lemma \ref{lem:Cnkab}, we need to treat the $n=0$ and $n=1$ cases separately.  Continuing equation \eqref{eq:LS+expand}, we arrive to

\begin{align}
    \label{eq:LS+switch}
    \P(LS+ \textrm{ succeeds}) =& \P(d(s,h) = 0) + (1-p)\P(d(s,h)=1)+ \nonumber \\
    &\sum_{n = 2}^{\infty} \sum_{\alpha,\beta \in \{0,1\}}
    \sum_{ k \in \mathcal{S}(n,\alpha,\beta) }
    \P(LS+ \textrm{ succeeds} | P \in C_{n,k, \alpha, \beta}) \frac{ |C_{n,k, \alpha, \beta}|}{ |C_{n}|} \P(d(s,h)=n)
\end{align}
Substituting in the results from Lemmas \ref{lem:Cn}, \ref{lem:Cnkab} and \ref{lem:LSp|Cnkab} into equation \eqref{eq:LS+switch} gives the desired result.
\end{proof}

\subsection{Proof of Theorem \ref{theorem:a}}
\label{sec:theorem:a}

We start by restating Theorem \ref{theorem:a} for convenience.

\begin{theorem}
\label{theorem:a_app}
In the $(d_r, d)$-RET with parameters $p_i,p_a, p_h$, let $a_{t,l}$ be  as in Definition \ref{def:total}. Then
\begin{align}
    a_{t,0} &= 1 \\
    a_{t,l} &= d_rp_i\sum_{m = l-1}^{t-1} \binom{m}{l-1}(1-p_i)^{m-l+1}d^{l-1}p_i^{l-1} \textrm{, for } t \geq l \geq 1 \\
    a_{t,l} &= 0 \textrm{, for l > t}.
\end{align}
\end{theorem}

\begin{proof}
Similarly to \cite{feng2018profile,mahmoud2021profile}, the proof relies on generating functions. We start by addressing the boundary cases. For all $t \geq 0$, it holds that $A_{t,0} = 1$, and therefore $a_{t,0} = 1$. Similarly, for all $l,t$ such that $l > t$, it holds that $A_{t,l} = 0$, and therefore $a_{t,l} = 0$. Suppose that $t\geq l = 1$. During day~$t-1$, on the first level, there are $A_{t-1,1}$ infected (internal) nodes and $d_r-A_{t-1,1}$ (external) nodes that may be infected with probability $p_i$ during day~$t$. Thus,
\begin{align}
\label{eq:At1}
    A_{t,1} &= A_{t-1,1}+\mathrm{Bin}(d_r - A_{t-1, 1};p_i) .
\end{align}
Taking the expectation of both sides in equation \eqref{eq:At1} yields
\begin{align}
    a_{t,1} = a_{t-1,1}(1-p_i) + d_rp_i, \textrm{ for }t\geq 1.
\end{align}
By subtracting the appropriate recurrence equations for $a_{t,1}$ and $a_{t-1,1}$ for $t\geq 2$ we obtain the homogeneous recurrence equation
\begin{align}
    a_{t,1} - a_{t-1,1}(2-p_i) + (1-p_i)a_{t-2,1} = 0, \textrm{ for } t\geq 2
\end{align}
and boundary conditions $a_{0,1} = 0$ and $a_{1,1} = d_rp_i$.
We solve for $a_{t,1}$ using the same methods as in the proof of Lemma~\ref{lem:Cn} and obtain
\begin{align}
\label{eq:a_t,1}
    a_{t,1} = d_r\left(1-(1-p_i)^t\right), \textrm{ for }t\geq 0.
\end{align}

Next, let us consider the general case $t \geq l > 1$. On day $t-1$, there are $A_{t-1, l-1}$ nodes on level $l-1$. Since, each node on level $l-1$ has $d$ children, there are $dA_{t-1,l-1}$ nodes on level $l$ that have an infectious parent on level $l-1$. However, $A_{t-1,l}$ of them are already infected. Therefore $dA_{t-1,l-1}-A_{t-1,l}$ nodes of level $l$ may be infected on day~$t$, each with probability $p_i$, which implies
\begin{align}
    \label{eq:Arec}
    A_{t,l} &= A_{t-1,l}+\mathrm{Bin}(dA_{t-1, l-1} - A_{t-1; l},p_i), \textrm{ for }t \geq l \geq 2 .
\end{align}
Taking the expectation of both sides in equation \eqref{eq:Arec} yields
\begin{align}
    \label{eq:atl_rec}
    a_{t,l} &= a_{t-1,l} + (da_{t-1, l-1}-a_{t-1, l})p_i \nonumber \\
    &= a_{t-1,l}(1-p_i) + dp_ia_{t-1, l-1}, \textrm{ for } t\geq l \geq 2.
\end{align}
For convenience, let us introduce $\lambda = 1-p_i$ and $\mu = dp_i$, and also let 
\begin{align}
\label{eq:f(x,y)}
    f(x,y) = \sum_{t = 1}^{\infty} \sum_{l = 1}^{\infty} a_{t,l}x^ty^l =\sum_{t = 1}^{\infty} \sum_{l = 1}^{t} a_{t,l}x^ty^l
\end{align}
be the generating function for $a_{t,l}$ with $t,l \geq 1$. By multiplying \eqref{eq:atl_rec} by $x^ty^l$ and summing it over $t,l \geq 2$ we obtain
\begin{align}
    \label{eq:a1}
    \sum_{t = 2}^{\infty} \sum_{l = 2}^{t} a_{t,l}x^ty^l &= \lambda \sum_{t = 2}^{\infty} \sum_{l = 2}^{t} a_{t-1,l}x^ty^l +\mu \sum_{t = 2}^{\infty} \sum_{l = 2}^{t} a_{t-1,l-1}x^ty^l \nonumber \\
     &= \lambda x\sum_{t = 1}^{\infty} \sum_{l = 2}^{t} a_{t,l}x^ty^l +\mu xy \sum_{t = 1}^{\infty} \sum_{l = 1}^{t} a_{t,l}x^ty^l.
\end{align}
Since $a_{1,l} = 0$ for $l \geq 2$, 
\begin{align}
    \label{eq:a2}
    \sum_{t = 1}^{\infty} \sum_{l = 2}^{t} a_{t,l}x^ty^l = \sum_{t = 2}^{\infty} \sum_{l = 2}^{t} a_{t,l}x^ty^l,
\end{align}
and by inserting \eqref{eq:a2} into \eqref{eq:a1}, we obtain
\begin{align}
    \label{eq:atl}
    (1-\lambda x)\sum_{t = 1}^{\infty} \sum_{l = 2}^{t} a_{t,l}x^ty^l &=  \mu xy \sum_{t = 1}^{\infty} \sum_{l = 1}^{t} a_{t,l}x^ty^l \stackrel{\eqref{eq:f(x,y)}}{=} \mu xy f(x,y) .
\end{align}
Now, we can also decompose the sum (\ref{eq:a2}) using geometric series as
\begin{align}
    \label{eq:atl2}
    \sum_{t = 1}^{\infty} \sum_{l = 2}^{t} a_{t,l}x^ty^l &= \sum_{t = 1}^{\infty} \sum_{l = 1}^{t} a_{t,l}x^ty^l - \sum_{t = 1}^{\infty} a_{t,1}x^ty \nonumber \\
    &\stackrel{\eqref{eq:a_t,1}}{=} f(x,y) - d_r y \sum_{t=1}^{\infty} (1-\lambda^t)x^t \nonumber \\
    &= f(x,y) - d_r xy \left( \frac{1}{1-x} - \frac{\lambda}{1-\lambda x}\right).
\end{align}
By plugging \eqref{eq:atl2} into \eqref{eq:atl}, we obtain the expression
\begin{align}
    \label{eq:f_x_y_final_expression}
    f(x,y) &= d_r (1-\lambda) xy\frac{1}{1-x}\frac{1}{1-\lambda x - \mu xy}.
\end{align}
Then, we expand the fractions in (\ref{eq:f_x_y_final_expression}) into a power series and we next apply the binomial theorem, we arrive to
\begin{align}
\label{eq:fnmj}
    f(x,y) &= d_r(1-\lambda)xy \sum_{n=0}^{\infty}x^n \sum_{m=0}^{\infty}x^m(\lambda+\mu y)^m \nonumber \\
    &= d_r(1-\lambda)xy \sum_{n=0}^{\infty}x^n \sum_{m=0}^{\infty}x^m\sum_{j = 0}^m \binom{m}{j}\lambda^{m-j} (\mu y)^{j} \nonumber \\
    &= d_r(1-\lambda)\sum_{n=0}^{\infty} \sum_{m=0}^{\infty}\sum_{j = 0}^m \binom{m}{j}\lambda^{m-j} \mu ^{j} x^{1+n+m}y^{j+1}.
\end{align}

Let $t = 1+n+m$ and $l= j+1$. In order to obtain an expression for $a_{t,l}$, we must change the variables in the sums of equation \eqref{eq:fnmj} from $(n,m,k)$ to $(t,m,l)$. Changing the inner sum from variable $j$ to $l$ is simple. Changing the variables in the two outer sums is more challenging because $t,n$ and $m$ depend on each other in a nontrivial way. More precisely, since $m,n\ge 0$ we have $t\ge 1$ and also $m\le t-1$, which means that we have to set the lower limit of $t$ and the upper limit of $m$ accordingly. As for the remaining limits, variable $t$ can be arbitrary large, and $m$ can take any integer value starting from $0$ independently of $t$, which yields the expression
\begin{align}
\label{eq:ftml}
    f(x,y) &= d_r(1-\lambda)\sum_{t=1}^{\infty} \sum_{m=0}^{t-1}\sum_{l = 1}^{m+1} \binom{m}{l-1}\lambda^{m-l+1} \mu ^{l-1} x^{t}y^{l} .
\end{align}
For the values of $l$ with $l \geq m+1$, the binomial coefficient $\binom{m}{l-1}$ is $0$, which implies that we can increase the upper limit of the inner sum from $m+1$ to $t$ in equation \eqref{eq:ftml}. Then, 
\begin{align}
\label{eq:ftml2}
     f(x,y) &= d_r(1-\lambda)\sum_{t=1}^{\infty} \sum_{m=0}^{t-1}\sum_{l = 1}^{t} \binom{m}{l-1}\lambda^{m-l+1} \mu ^{l-1} x^{t}y^{l} \nonumber \\
     &= \sum_{t=1}^{\infty} \sum_{l = 1}^{t} d_r(1-\lambda) \sum_{m=0}^{t-1} \binom{m}{l-1}\lambda^{m-l+1} \mu ^{l-1} x^{t}y^{l}.
\end{align}
Finally we can read off the value of $a_{t,l}$ from equation \eqref{eq:ftml2} as 
\begin{align}
    a_{t,l} = d_r(1-\lambda) \sum_{m=0}^{t-1} \binom{m}{l-1} \mu^{l-1} \lambda^{m-l+1} = d_rp_i \sum_{m=0}^{t-1} \binom{m}{l-1} (dp_i)^{l-1} (1-p_i)^{m-l+1}.
\end{align}

\end{proof}

\subsection{Proof of Corollary \ref{corollary:at}}
\label{sec:corollary:at}

We start by restating Corollary \ref{corollary:at} for convenience.

\begin{corollary}
\label{corollary:at_app}
In the RET$(p_i, d_r, d)$, let $a_{t}$ be the expectation of \eqref{eq:def_total}, as in Definition~\ref{def:total}. For $t\geq0$, 
\begin{equation}
    a_{t} = 1 + d_r\frac{(1-p_i+dp_i)^t - 1}{d-1}. \label{lemma:total_app}
\end{equation}
\end{corollary}

\begin{proof}
By using linearity of expectation, equation \eqref{eq:def_total} and Theorem \ref{theorem:a} we obtain:
\begin{align}
    \label{eq3:1}
    a_{t} &= \sum_{l = 0}^{+\infty} a_{t,l} \nonumber \\
    &= 1 + \sum_{l = 1}^{+\infty} a_{t,l} \nonumber \\
    &= 1 + d_rp_i\sum_{l = 1}^{t} \sum_{m = l-1}^{t-1} \binom{m}{l-1}(1-p_i)^{m-l+1}d^{l-1}p_i^{l-1} 
\end{align}
Before we use binomial theorem, we need to swap the sums. Boundaries from \eqref{eq3:1} are equivalent to $t-1 \geq m \geq l-1 \geq 0$, so we can rewrite this as 2 conditions: $m+1 \geq l \geq 1$ and $t \geq m \geq 0$. 
\begin{align}
    a_{t,l}
    &= 1 + d_rp_i\sum_{m = 0}^{t-1}\sum_{l = 1}^{m+1} \binom{m}{l-1}(1-p_i)^{m-l+1}d^{l-1}p_i^{l-1} \nonumber \\
    &= 1 + d_rp_i\sum_{m = 0}^{t-1}\sum_{l = 0}^{m} \binom{m}{l}(1-p_i)^{m-l}d^{l}p_i^{l} 
\end{align} 
Finally, by applying the binomial theorem and summing the geometric series, we obtain the desired equation:
\begin{align}
    a_{t,l}
    &= 1 + d_rp_i\sum_{m = 0}^{t-1} (1-p_i+dp_i)^m \nonumber \\ 
    &= 1 + d_r\frac{(1-p_i+dp_i)^t - 1}{d-1}.
\end{align}
\end{proof}

\subsection{Proof or Lemma \ref{lem:DET}}
\label{sec:lem:DET}
We restate Lemma \ref{lem:DET} here for convenience.

\begin{lemma}
\label{lem:DET_app}
Let us consider the stopped DET model with parameters $(c_{t,l}), p_a, p_h$, and let $h$ denote the first hospitalized node. Then
\begin{align}
\label{eq:DET_lemma_app}
    \P(d(s,h) = l) = \sum_{t=0}^{+\infty} \frac{c_{t,l}-c_{t-1,l}}{c_t-c_{t-1}}  (1-(1-p_a)p_h)^{c_{t-1}}\left(1-(1-(1-p_a)p_h)^{c_t-c_{t-1}}\right).
\end{align}
\end{lemma}

\begin{proof}
Recall that a node added at day $t$ is uniformly distributed among the $c_{t}-c_{t-1}>0$ nodes added that day, and that the number of nodes added to level $l$ is $c_{t,l}-c_{t-1,l}$ on day $t$. If we condition on the time of the first hospitalized case, denoted by $TI_h$, then
\begin{align}
\P(d(s,h) = l) &= \sum_{t=0}^{+\infty} \P(d(s,h)=l | TI_h = t) \P(TI_h = t)  \nonumber \\
& = \sum_{t=0}^{+\infty} \frac{c_{t,l}-c_{t-1,l}}{c_t-c_{t-1}}  \P(\textrm{node is not hosp})^{c_{t-1}} (1-\P(\textrm{node is not hosp })^{c_t-c_{t-1}})
\nonumber \\
& = \sum_{t=0}^{+\infty} \frac{c_{t,l}-c_{t-1,l}}{c_t-c_{t-1}}  (1-(1-p_a)p_h)^{c_{t-1}}\left(1-(1-(1-p_a)p_h)^{c_t-c_{t-1}}\right).
\label{eq:DET_ctl}
\end{align}
\end{proof}


\section{Dynamic Message Passing for the DDE model}
\label{sec:DMP_all}

In this section, we explain how we derived and implemented the DMP equations for the DDE+HNM model. We start by reviewing the previous work on the DMP equations for the SIR model in Appendix \ref{sec:DMP_SIR}, and then we proceed to our derivations in Appendix \ref{sec:DMP_DDE}. In Appendix \ref{sec:DMP_feasible}, we explain how we find candidate (node,time) pairs for the DMP equations, and in Appendix \ref{sec:DMP_final} we conclude by combining Appendices \ref{sec:DMP_DDE} and \ref{sec:DMP_feasible} into a source-detection algorithm.

\subsection{DMP Equations for the SIR Model}
\label{sec:DMP_SIR}

The DMP equations were first derived by \cite{lokhov2014inferring} for the SIR model in the context of source detection. Their goal is to compute the marginal probabilities that node $i$ is in a given state at time $t$ (denoted by $P_S^{i}(t), P_I^{i}(t)$ and $P_R^{i}(t)$ for the susceptible, infected and recovered states, respectively), given initial conditions $P_S^{i}(t_0), P_I^{i}(t_0)$ and $P_R^{i}(t_0)$ at some initial time $t_0$. To solve this problem in tree networks, we may consider a dynamic programming approach, where we delete a node $i$, we compute the marginal probabilities of $P_S^{j}(t-1)$ for all neighbors $j$ of $i$ in the remaining subtrees, and use this information to compute $P_S^{i}(t)$ (as the marginals are independent in each of the subtrees conditioned on the state of $i$). The DMP equations make the dynamic programming intuition explicit. Originally, the DMP equations were developed for static networks, but since the generalization to time-varying networks is straightforward, and has already been foreshadowed in a similar heuristic algorithm \cite{jiang2016rumor}, we include it in this preliminary section. For time-varying networks, we define $N_i(t)$ as the set of neighbors of node $i$ in the time-window $[t,t+1)$.

To formalize the dynamic programming approach, \cite{lokhov2014inferring} introduces some new notation. Let $\lambda$ be the probability that an infectious node infects a  susceptible neighbor, and let $\mu$ be the probability that an infectious node recovers. Let $D_{i}$ be the auxiliary dynamics, where node $i$ receives infection signals, but ignores them, and thus remains in the $S$ state at all times. Let $P_{S}^{j \rightarrow i}(t)$ be the probability that node $j$ is in the state $S$ at time $t$ in the dynamics $D_{i}$, and let $\theta^{k \rightarrow i}(t)$ be the probability that the infection signal has not been passed from node $k$ to node $i$ up to time $t$ in the dynamics $D_{i}$. Finally, let $\phi^{k \rightarrow i}(t)$ be the probability that the infection signal has not been passed from node $k$ to node $i$ up to time $t$, and that node $k$ is in the state $I$ at time $t$, in the dynamics $D_{i}$. With these definitions, the dynamic programming approach is formalized by the following equations for $t\ge t_0$:

\begin{align}
 P_S^{i \rightarrow j}(t+1)&=P_S^{i}(t_0)\prod_{k\in N_i(t) \backslash j}\theta^{k \rightarrow i}(t+1), \label{eq:SIRequations:Ps}
\\
\theta^{k \rightarrow i}(t+1)-\theta^{k \rightarrow i}(t) &=
-\lambda\phi^{k \rightarrow i}(t), \label{eq:SIRequations:theta}
\\
\phi^{k \rightarrow i}(t)&=(1-\lambda)(1-\mu)\phi^{k \rightarrow i}(t-1) + \left(P_S^{k \rightarrow i}(t-1)-P_S^{k \rightarrow i}(t)\right). \label{eq:SIRequations:phi}
\end{align}

The marginal probabilities that node $i$ is in a given state at time $t$ are then given by
\begin{align}
& P_S^{i }(t+1)=P_S^{i}(t_0)\prod_{k\in N_i(t)}\theta^{k \rightarrow i}(t+1)\, ,\label{eq:SIRequations:S}
\\
& P_R^{i}(t+1)=P_R^{i}(t)+\mu P_{I}^{i}(t)\, ,\label{eq:SIRequations:R}
\\
& P_I^{i}(t+1)=1-P_S^{i}(t+1)-P_R^{i}(t+1)\, .\label{eq:SIRequations:I}
\end{align}

These equations are only exact on trees, but they can also be applied to networks with cycles as a heuristic approach. The heuristic gives good approximations to the true marginals if the network is at least locally tree-like \cite{karrer2010message}.

\subsection{DMP Equations for the DDE+HNM Model}
\label{sec:DMP_DDE}

There are several differences between the SIR model on locally tree-like networks and the DDE+HNM model (see Figure \ref{fig:mobility_models} (a)). First, the DDE model has additional compartments (exposed nodes, asymptomatic nodes), which motivates the introduction of several new variables. Let $\lambda_{(a)}$ (resp., $\lambda_{(s)}$) be the probability that an asymptomatic (resp., symptomatic) node infects a susceptible node. Let $\phi^{k \rightarrow i}(t)^{(a)}$ (resp., $\phi^{k \rightarrow i}(t)^{(s)}$) be the probability that the infection signal has not been passed from node $k$ to node $i$ up to time $t$, and that node $k$ is asymptomatic (resp., symptomatic) infectious at time $t$, in the dynamics $D_{i}$.

The second important difference is that in the DDE model, the transition times between different compartments are deterministic instead of following a geometric distribution as in the standard SIR model. While deterministic transition times sound simpler at first, it turns out that they make the DMP equations more complex, because the Markovian property that each marginal probability depends only on the previous timestep is lost if the transition times are larger than $1$. Recall that the times for the transitions $E \rightarrow I$ and $I \rightarrow R$ (with their default values) are $T_E=3$ and $T_I=14$.

Let us incorporate these two differences into equations \eqref{eq:SIRequations:Ps}--\eqref{eq:SIRequations:phi} to derive the DMP equations for the DDE model. Equation \eqref{eq:DDEequations:Ps} is essentially a copy of \eqref{eq:SIRequations:Ps}. Equation \eqref{eq:DDEequations:theta} follows equation \eqref{eq:SIRequations:theta}, but we incorporate the two different variants of infected (asymptomatic and symptomatic) patients with their respective infection probabilities $\lambda_{(a)}$ and $\lambda_{(s)}$. Equation \eqref{eq:DDEequations:Pr} is a new equation, which is necessary because recovery times are no longer geometric random variables; instead we need to check the probabilities of infection $T_E+T_I$ timesteps earlier than the current time $t$. Finally, equation \eqref{eq:DDEequations:phi_a} (resp., \eqref{eq:DDEequations:phi_s}) is the asymptomatic (resp., symptomatic) version of equation \eqref{eq:SIRequations:phi}, while also incorporating the deterministic time for the transition $E \rightarrow I$. For $t\ge t_0$, this yields equations

\begin{align}
 P_S^{i \rightarrow j}(t+1)&=P_S^{i}(t_0)\prod_{k\in N_i(t) \backslash j}\theta^{k \rightarrow i}(t+1)=P_S^{i}(t_0)\frac{P_S^{i}(t+1)}{\theta^{j \rightarrow i}(t+1)}, \label{eq:DDEequations:Ps}
\\
\theta^{k \rightarrow i}(t+1)-\theta^{k \rightarrow i}(t) &=
-\lambda_{(a)} \phi^{k \rightarrow i}_{(a)}(t)-\lambda_{(s)}\phi^{k \rightarrow i}_{(s)}(t),
\label{eq:DDEequations:theta}
\\
P_R^{k \rightarrow i}(t) &=P_S^{k \rightarrow i}(t-T_E-T_I-1)-P_S^{k \rightarrow i}(t-T_E-T_I) \label{eq:DDEequations:Pr}
\\
\phi^{k \rightarrow i}_{(a)}(t)&=(1-\lambda_{(a)})(1-P_R^{k \rightarrow i}(t)) \phi^{k \rightarrow i}_{(a)}(t-1) \nonumber\\
& \qquad \qquad + p_a[P_S^{k \rightarrow i}(t-T_E-1)-P_S^{k \rightarrow i}(t-T_E)]. \label{eq:DDEequations:phi_a}
\\
\phi^{k \rightarrow i}_{(s)}(t)&=(1-\lambda_{(s)})(1-P_R^{k \rightarrow i}(t)) \phi^{k \rightarrow i}_{(s)}(t-1) \nonumber\\
& \qquad \qquad + (1-p_a)[P_S^{k \rightarrow i}(t-T_E-1)-P_S^{k \rightarrow i}(t-T_E)]. \label{eq:DDEequations:phi_s}
\end{align}

We note that for early values of $t$, equations \eqref{eq:DDEequations:Pr}--\eqref{eq:DDEequations:phi_s} depend on $P_S^{k \rightarrow i}$ before $t_0$, which we initialize to be 1 (all nodes are susceptible before the first node develops the infection). The marginal probability that node $i$ is susceptible at time $t$ is still computed by equation \eqref{eq:SIRequations:S} as before. Equations \eqref{eq:SIRequations:R}--\eqref{eq:SIRequations:I} do not apply anymore; we explain it in Appendix \ref{sec:DMP_final} how to take into account observations for nodes in the infectious compartments.

The third difference between the the SIR model on locally tree-like networks and the DDE+HNM model is that the HNM model contains many short cycles inside the households. Short cycles can cause unwanted feedback loops in the DMP equations where, loosely speaking, nodes are treated as if they could reinfect themselves. We solve this issue by modifying the underlying graph to be locally tree-like (only for the computation of the DMP equations). Specifically, we introduce a new central household-node for each household, and we replace the cliques inside the households by a star graph centered at this new household-node node. Introducing such a central household-node does of course alter epidemic process, in particular it makes household infections less independent and slower (all household infections need to pass through an extra node). To mitigate this issue, we assume that central household-nodes have $T_E=1$ and that they are infected with probability 1 by any node in the same household. We tested the validity of the resulting DMP equations against simulations of the epidemic progressions and we found the results to be quite accurate, in particular, more accurate than the version without the introduction of these central household-nodes.

Note that we derived the DMP equations for the DDE+HNM model, however, since (i) the compartments are the same, (ii) the equations support temporal networks, and (iii) we have separate infection probabilities $\lambda_{(a)}$ and $\lambda_{(s)}$ for asymptomatic and symptomatic nodes, our equations can also be applied to the DCS+TU model after a discretizing (rounding) the time observations.

Finally, we touch upon the computational complexity of computing the DMP equations. In principle, we need to update $O(dN)$ equations (for each edge) over $t_{\mathrm{max}}$ timesteps, where $t_{\mathrm{max}}$ is the maximum time during which the marginals can still change, which can be as large as $O(N)$. However, since we are only interested in computing the likelihood of the 5 earliest observations, $t_{\mathrm{max}}$ is typically quite low. Moreover, since we assume to be in an early stage of the epidemic, most of the equations remain unchanged. For better computational scalability, we only compute $P_S^{i}(t)$ and $\theta^{k \rightarrow i}(t)$ for nodes $k,i$ that have $P_I^{k \rightarrow i}(t)>0.01$, i.e., we only update nodes that are at least somewhat likely to have received the infection. Otherwise, we set $P_I^{k \rightarrow i}(t)=P_I^{k \rightarrow i}(t-1)$, $\theta^{k \rightarrow i}(t)=\theta^{k \rightarrow i}(t-1)$, and in the implementation we can perform these assignments implicitly using appropriate data structures. With these adjustments, the time-complexity of the algorithm becomes independent of $N$, but remains dependent on the network parameters, the epidemic parameters and the number of sensors in a non-trivial way.

\subsection{Feasible Source-time Pairs for Source Detection}

\label{sec:DMP_feasible}

In this section we explain how we implemented the feasible source identification algorithm, which was suggested as a preprocessing step for a method very similar to the DMP equations by \cite{jiang2016rumor}. Let us define the directed graph $G_2$ on (node,infecton\_time) pairs (we use ``nodes'' for the nodes of the original graph $G$ and ``pairs'' for the nodes of $G_2$), and draw an edge between two pairs $(v_1, t_1) \rightarrow (v_2, t_2)$ if $v_1$ and $v_2$ are in contact at $t_2$, and $t_2$ is in the interval $[t_1+T_E, t_1+T_E+T_I]$. Observe that in the DDE model there is an edge $(v_1, t_1) \rightarrow (v_2, t_2)$ if and only if $v_1$ becoming infected at time $t_1$ can infect $v_2$ at time $t_2$. The definition of $G_2$ is applicable to the DCS model as well after discretization (rounding), however, since the infection times are not deterministic anymore, not all possible infections $(v_1, t_1) \rightarrow (v_2, t_2)$ have a corresponding edge in~$G_2$.

Then, we perform a breadth-first search backwards on the directed edges of $G_2$, starting from each pair $(v_i, t_i-T_E-T_P)$, where $v_i$ is a symptomatic sensor node, and $t_i$ is the symptom onset time of $v_i$ (for the DCS model, we start from integer times in the $t_i-T_E-T_P \pm (\sigma_E+\sigma_P)$ interval to account for the randomness of the transition times). To limit the time complexity of the algorithm, we only consider the $k_1$ earliest observations, which means that we start $k_1$ breadth-first searches. 
With this construction, each pair $(v,t)$ discovered by a breadth-first search started from $(v_i, t_i-T_E-T_P)$ could have caused the infection in $v_i$; we say that $(v,t)$ is an explanation for observation $i$. We perform the breath-first searches until we find $k_2$ pairs that explain all of the $k_1$ earliest observations. See the pseudocode in Algorithm \ref{alg:feasible}.

\begin{claim}
In the DDE model, Algorithm \ref{alg:feasible} with $\sigma_E=\sigma_P=0$ finds the $k_2$ feasible explanations with the latest starting time of the $k_1$ earliest symptomatic nodes.
\end{claim}
\begin{proof}
By construction, a source node $v$ that becomes infectious at time $t$ can cause an observation $(v_i, t_i)$ if and only if there is a directed path from $(v,t)$ to $(v_i, t_i-T_E-T_P)$. Therefore, the breadth-first search algorithm finds all of the closest feasible sources in time.
\end{proof}

\begin{algorithm}[h]
    \SetKwFunction{isOddNumber}{isOddNumber}

    \KwIn{
    \begin{itemize}
    \item The mean exposed time $T_E$ the mean pre-infectious time $T_P$, the mean infectious time $T_I$,\\ the std of the exposed time $\sigma_E$ and the std of the pre-infectious time $\sigma_P$
    \item  $F(v)_{min}$ and $F(v)_{max}$ returns the minimum and maximum times when $v$ could have been exposed based on all of its (possibly asymptomatic or negative) test results
    \item  $S(v)$ returns the time of symptom onset for a node $v$ tested positive symtomatic.
    \item $N(v, [ t_{min}, t_{max} ])$ returns the set of neighbors of node $v$ in the interval $[t_{min}, t_{max}]$
    \item  A lower estimate of the time the source became infectious $t_{min}$
    \item Integers $k_1,k_2$
    \end{itemize}}
    \KwOut{A list of at most $k_2$ tuples of node and time pairs that can explain the first $k_1$ symptomatic nodes}
    $l \leftarrow \{\}$; \tcp*[f]{if the list $l[t]$ contains the tuple $(v,w)$, then the infection started at $w$ at time $t$ can explain $v$}\\
    $D \leftarrow \{\}$; \tcp*[f]{if the list $D[w,t]$ contains the node $v$, then the infection started at $w$ at time $t$ can explain $v$}\\
    $doneList \leftarrow []$\;
    \For{$v \in SortIncreasingByValues(S)[0:k_1]$} {
        $t'_{min} \leftarrow S(v)-(T_E+T_P)-(\sigma_E+\sigma_P)$\;
        $t'_{max} \leftarrow S(v)-(T_E+T_P)+(\sigma_E+\sigma_P)$\;
        \For{$t' \leftarrow t'_{min}$  to $t'_{max}$} {
            $Append((v,v), l[t'])$\;
             $Append(v, D[v,t'])$\;
            \uIf{$Length(D[v,t'])=k_1$} {
                $Append((v,t'),doneList)$
            }        
        }
     }
     $t \leftarrow SortIncreasingByValues(S)[k_1-1]$\;
     $stopCondition \leftarrow False$ \;
     \While{ not $stopCondition$ and  $t>t_{min}$} {
         \For {$v,w \in l[t]$} {
            \For{$u \in N(w, [t,t-1])$} {
                 $t'_{min} \leftarrow \max(F(u)_{min}, t-T_E-T_I)$\;
                 $t'_{max} \leftarrow \min(F(u)_{max}, t-T_E)$\;
                \For{$t' \leftarrow t'_{max}$  to $t'_{min}$} {
                    $Append((v,u), l[t'])$\;
                   $ Append(v,D[(u,t')])$\;
                   \uIf{$Length(D[u,t'])=k_1$} {
                     $Append((u,t'),doneList)$
                   }
                }
            }
         }
         $doneList \leftarrow SortBySecondElement(doneList)$\;
         \uIf{$Length(doneList\ge k2$) and $t-T_E \le doneList[k2][1]$} {
             $stopCondition \leftarrow True$\;
         }
         \Else{
             $t \leftarrow t-1$\;
         }
    }
    \KwRet{$doneList$}
    \caption{Feasible source identification (reverse dissemination \cite{jiang2016rumor})}
    \label{alg:feasible}
\end{algorithm}

\subsection{Source Detection via Feasible Source Identification and DMP}

\label{sec:DMP_final}

In this section we explain how to combine Algorithm~\ref{alg:feasible} with the DMP equations derived in Appendix~\ref{sec:DMP_DDE}. See the pseudocode in Algorithm~\ref{alg:Sdd}.

We start by computing the DMP equations \eqref{eq:DDEequations:Ps}-\eqref{eq:DDEequations:phi_s} and \eqref{eq:SIRequations:Ps} for the $k_2$ tuples of node and time pairs that can explain the first $k_1$ symptomatic observations returned by Algorithm \ref{alg:feasible}. Next, our goal is to use these DMP equations to compute the likelihood of each of the $k_2$ tuples using the $k_1$ observations. Similarly to \cite{lokhov2014inferring}, we make the assumption that the first $k_1$ observations are independent, and we can compute the likelihood by multiplying their respective marginals together. For symptomatic observed nodes $v$, we know the time of symptom onset, which we denote by $S(v)$. Then, the marginal probability of $v$ developing symptoms exactly at time $t$ can be computed by taking the difference of $P_S^{v}(S(v)-T_P-T_E-1)$ and $P_S^{v}(S(v)-T_P-T_E)$ and multiplying the difference by $(1-p_a)$. In Algorithm \ref{alg:Sdd} we drop the multiplicative factor $(1-p_a)$ because it is present for all of the tuples, and it does not change the final order of their scores. For asymptomatic (resp., negative) observations, we only know that at the time of testing, denoted by $A(v)$ (resp., $NE(v)$), at least a time interval of length~$T_E$ has passed (resp., $T_E$ has not passed) since the time of infection. Therefore, dropping the $p_a$ factor similarly to the symptomatic case, we compute the marginal of asymptomatic observations as $1-P_S^{v}(A(v)-T_E)$, and we compute the marginal of negative observations as $P_S^{v}(NE(v)-T_E)$. Finally, the contributions of the observations are multiplied together for each of the $k_2$ tuples returned by Algorithm \ref{alg:feasible}, and the scores approximating the likelihoods are returned.

%

\begin{algorithm}
    \SetKwFunction{isOddNumber}{isOddNumber}

    \KwIn{
    \begin{itemize}
    \item The mean exposed time $T_E$, the mean pre-infectious time $T_P$, the mean infectious time $T_I$
    \item  $S(v)$ returns the time of symptom onset for a node $v$ tested positive symtomatic.
    \item  $A(v)$ and $NE(v)$ return the time of asymptomatic and negative test results, respectively
    \end{itemize}}
    \KwOut{A dictionary $L$ of $k_2$ elements, which contains a score for each $(v,t)$ pair that explains the first $k_1$ observations. Higher scores signify higher confidence of being the source.}
    $L \leftarrow \{\}$\;
    $doneList \leftarrow \mathrm{Algorihtm \  \ref{alg:feasible}}(k_1,k_2)$\;

    \For{$v,t_0 \in doneList$} {
        $P_S  \leftarrow$ eq. \eqref{eq:SIRequations:Ps} based on DMP eq. \eqref{eq:DDEequations:Ps}-\eqref{eq:DDEequations:phi_s} with $P_S^v(t_0)=0$, and $P_S^w(t_0)=1$ for all $w \ne v$\;
        $L[v,t_0] \leftarrow 1$\;
        \For{$w \in S$} {
            $L[v,t_0] \leftarrow L[v,t_0]\cdot (P_S^{v}(S(v)-T_P-T_E-1)-P_S^{v}(S(v)-T_P-T_E))$\;
        }
      \uIf{{$w \in A$} }{
                 $L[v,t_0] \leftarrow L[v,t_0]\cdot (1-P_S^{v}(A(v)-T_E))$\;
        }
        \For{$w\in NE$} {
                 $L[v,t_0] \leftarrow L[v,t_0]\cdot P_S^{v}(NE(v)-T_E)$\;
        }
    }
    \KwRet{$L$}
    \caption{Source detection via DMP}
    \label{alg:Sdd}
\end{algorithm} 

\end{document}